\documentclass[draftclsnofoot, onecolumn, journal]{IEEEtran}
\usepackage{amsmath}
\usepackage{amssymb}
\usepackage[dvips]{graphicx}
\usepackage{color}
\usepackage{subfigure}

\newtheorem{RK}{Remark}

\newtheorem{define}{Definition}

\newtheorem{lem}{Lemma}
\newtheorem{thm}{Theorem}

\newtheorem{cor}{Corollary}

\newcommand{\nn}{\nonumber}
\newcommand{\mc}{\mathcal}

\newcommand{\ms}{\mathsf}
\newcommand{\mr}{\mathrm}

\makeatletter

\newcommand{\be}{\begin{equation}}
\newcommand{\ee}{\end{equation}}

\newcommand{\bc}{\begin{center}}
\newcommand{\ec}{\end{center}}
\newcommand{\bfl}{\begin{flushleft}}
\newcommand{\efl}{\end{flushleft}}

\newcommand{\beqa}{\begin{eqnarray}}
\newcommand{\eeqa}{\end{eqnarray}}
\newcommand{\beqan}{\begin{eqnarray*}}
\newcommand{\eeqan}{\end{eqnarray*}}
\newcommand{\beq}{\begin{equation}}
\newcommand{\eeq}{\end{equation}}

\newcommand{\lp}{\left (}
\newcommand{\rp}{\right )}

\begin{document}
%
% paper title
% can use linebreaks \\ within to get better formatting as desired
\title{{On the Sum-Capacity with Successive Decoding in Interference Channels}}

\author{~\\
\IEEEauthorblockN{Yue Zhao, Chee Wei Tan, A. Salman Avestimehr, Suhas N. Diggavi, Gregory J. Pottie}
}

% make the title area
\maketitle

%\markboth{Submitted to IEEE Transactions on Information Theory}{}

\begin{abstract}
%\boldmath
In this paper, we investigate the sum-capacity of the two-user Gaussian interference channel with Gaussian superposition coding and successive decoding. We first examine an approximate deterministic formulation of the problem, and introduce the complementarity conditions that capture the use of Gaussian coding and successive decoding. In the deterministic channel problem, we find the constrained sum-capacity and its achievable schemes with the minimum number of messages, first in symmetric channels, and then in general asymmetric channels. We show that the constrained sum-capacity \emph{oscillates} as a function of the cross link gain parameters between the information theoretic sum-capacity and the sum-capacity with interference treated as noise. Furthermore, we show that if the number of messages of either of the two users is fewer than the minimum number required to achieve the constrained sum-capacity, the maximum achievable sum-rate drops to that with interference treated as noise. We provide two algorithms (a simple one and a finer one) to translate the optimal schemes in the deterministic channel model to the Gaussian channel model. We also derive two upper bounds on the sum-capacity of the Gaussian Han-Kobayashi schemes, which automatically upper bound the sum-capacity using successive decoding of Gaussian codewords. Numerical evaluations show that, similar to the deterministic channel results, the constrained sum-capacity in the Gaussian channels oscillates between the sum-capacity with Han-Kobayashi schemes and that with single message schemes.
\end{abstract}

\section{Introduction}
We consider the sum-rate maximization problem in two-user Gaussian interference channels (cf. Figure \ref{twoIC}) under the constraints of successive decoding. While the information theoretic capacity region of the Gaussian interference channel is still not known, it has been shown that a Han-Kobayashi scheme with random Gaussian codewords can achieve within 1 bit/s/Hz of the capacity region \cite{ETW08}, and hence within 2 bits/s/Hz of the sum-capacity. In this Gaussian Han-Kobayashi scheme, each user first decodes both users' common messages jointly, and then decodes its own private message. In comparison, the simplest commonly studied decoding constraint is that each user treats the interference from the other users as noise, i.e., without any decoding attempt. Using Gaussian codewords, the corresponding constrained sum-capacity problem can be formulated as a non-convex optimization of power allocation, which has an analytical solution in the two-user case \cite{two06}. It has also been shown that within a certain range of channel parameters for \emph{weak} interference channels, treating interference as noise achieves the information theoretic sum-capacity \cite{VVV09, Mota09, Shang09}. For general interference channels with \emph{more than two} users, there is so far neither a near optimal solution information theoretically, nor a polynomial time algorithm that finds a near optimal solution with interference treated as noise \cite{LZ08} \cite{TFL11}.

In this paper, we consider a decoding constraint --- \emph{successive decoding of Gaussian superposition codewords} --- that bridges the complexity between joint decoding (e.g. in Han-Kobayashi schemes) and treating interference as noise. %A comparison of the three schemes is given in Table \ref{compare3}.
We investigate the constrained sum-capacity and its achievable schemes. Compared to treating interference as noise, allowing successive cancellation yields a much more complex problem structure. To clarify and capture the key aspects of the problem, we resort to the deterministic channel model \cite{ADTjour}. In \cite{BT08}, the information theoretic capacity region for the two-user deterministic interference channel is derived as a special case of the El Gamal-Costa deterministic model \cite{El82}, and is shown to be achievable using Han-Kobayashi schemes.

We transmit messages using a superposition of Gaussian codebooks, and use successive decoding. To capture the use of successive decoding of Gaussian codewords, in the deterministic formulation, we introduce the \emph{complementarity conditions} on the bit levels, which have also been characterized using a conflict graph model in \cite{Shaoinfo}. We develop transmission schemes on the bit-levels, which in the Gaussian model corresponds to message splitting and power allocation of the messages. We then solve the constrained sum-capacity, and show that it \emph{oscillates} (as a function of the cross link gain parameters) between the information theoretic sum-capacity and the sum-capacity with interference treated as noise. Furthermore, the minimum number of messages needed to achieve the constrained sum-capacity is obtained. Interestingly, we show that if the number of messages is limited to even \emph{one less} than this minimum capacity achieving number, the sum-capacity drops to that with interference treated as noise.

We then translate the optimal schemes in the deterministic channel to the Gaussian channel, using a rate constraint equalization technique. To evaluate the optimality of the translated achievable schemes, we derive and compute two upper bounds on the sum-capacity of Gaussian Han-Kobayashi schemes\footnote{Throughout this paper, when we refer to the Han-Kobayashi scheme, we mean the Gaussian Han-Kobayashi scheme, unless stated otherwise.}. %  with joint decoding, unless stated otherwise.}.
Since a scheme using superposition coding with Gaussian codebooks and successive decoding is a special case of Han-Kobayashi schemes, these bounds automatically apply to the sum-capacity with such successive decoding schemes as well. We select two mutually exclusive subsets of the inequality constraints that characterize the Gaussian Han-Kobayashi capacity region. Maximizing the sum-rate with each of the two subsets of inequalities leads to one of the two upper bounds. The two bounds are shown to be tight in different ranges of parameters. Numerical evaluations show that the sum-capacity with Gaussian superposition coding and successive decoding oscillates between the sum-capacity with Han-Kobayashi schemes and that with single message schemes.

The remainder of the paper is organized as follows. Section \ref{probform} formulates the problem of sum-capacity with successive decoding of Gaussian superposition codewords in Gaussian interference channels, and compares it with Gaussian Han-Kobayashi schemes. Section \ref{DCsec} reformulates the problem with the deterministic channel model, and then solves the constrained sum-capacity. Section \ref{GCsec} translates the optimal schemes in the deterministic channel back to the Gaussian channel, and derives two upper bounds on the constrained sum-capacity. Numerical evaluations of the achievability against the upper bounds are provided. Section \ref{disc} concludes the paper with a short discussion on generalizations of the coding-decoding assumptions and their implications.

\section{Problem formulation in Gaussian channels} \label{probform}
We consider the two-user Gaussian interference channel shown in Figure \ref{twoIC}. The received signals of the two users are
%\beq y_i = \sum_{j=1}^K h_{ji} x_j + z_i, \quad i=1,2 \eeq
\begin{align}
& y_1 = h_{11}x_1 + h_{21}x_2 + z_1, \nn\\
& y_2 = h_{22}x_2 + h_{12}x_1 + z_2, \nn
\end{align}
where $\{h_{ij}\}$ are constant complex channel gains, and $z_i\sim\mc{C}\mc{N}(0,N_i)$. Define $g_{ij} \triangleq |h_{ij}|^2, (i,j = 1,2)$.

\begin{figure}[tb]
  \centering
  \includegraphics[scale = 0.7]{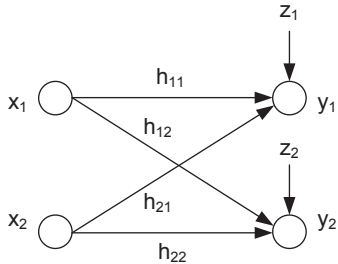}
  \caption{Two-user Gaussian interference channel.}
  \label{twoIC}
\end{figure}

There is an average power constraint equal to $\bar{p}_i$ for the $i^{th}$ user ($i = 1,2$). In the following, we first formulate the problem of finding the optimal Gaussian superposition coding and successive decoding scheme, and then provide an illustrative example to show that successive decoding schemes do not necessarily achieve the same capacity as Han-Kobayashi schemes.

\subsection{Gaussian Superposition Coding and Successive Decoding: a Power and Decoding Order Optimization}
Suppose the $i^{th}$ user uses a superposition of $L_i$ messages $x_i^{(\ell)} (1\leq \ell \leq L_i)$. Denote by $r_i^{(\ell)}$ the rate of message $x_i^{(\ell)}$. For a given block length $n$, for each message $x_i^{(\ell)}$, a codebook of size $2^{nr_i^{(\ell)}}$ is generated by using IID random variables of $\mc{C}\mc{N}(0,1)$. The codebooks for different messages are independently generated. For the $i^{th}$ user, the transmit signal $x_i$ is a superposition of $L_i$ Gaussian codewords, with its individual power constraint $\bar{p}_i$ satisfied, i.e.,
\begin{equation*}
x_i= \sum_{\ell=1}^{L_i} \sqrt{p_i^{(\ell)}} x_i^{(\ell)},
\end{equation*}
\beq \label{eq:PowerConstraints2} \sum_{\ell=1}^{L_i} p_i^{(\ell)} \leq \bar{p}_i, \quad i=1, 2.\eeq

The $i^{th}$ receiver attempts to decode all $x_i^{(\ell)}$, $\ell=1,\ldots,L_i$, using successive decoding as follows. It chooses a decoding order $\mc{O}_i$ of all the $L_1+L_2$ messages from both users. %($j=1,2$, $\ell=1,\ldots,L$).
It starts decoding from the first message in this order (by treating all other messages that are not yet decoded as noise,) then peeling it off and moving to the next one, until it decodes all the messages intended for itself --- $x_i^{(\ell)}$, $\ell=1,\ldots,L_i$.

Denote the message that has order $q$ in $\mc{O}_i$ by $x_{t_{q,i}}^{(\ell_{q,i})}$, i.e., it is the ${\ell_{q,i}}^{th}$ message of the ${t_{q,i}}^{th}$ user. Then, the achievable rate for the successive decoding procedure to have a vanishingly small error probability as the block length $n\rightarrow\infty$ yields the following constraints on the rates of the messages:
\beq  \label{eq:DecodConstraints} r_{t_{q,i}}^{(\ell_{q,i})} \leq  \log \lp 1+ \frac{p_{t_{q,i}}^{(\ell_{q,i})} g_{t_{q,i}i}}{\sum_{s=q+1}^{L_1+L_2} p_{t_{s,i}}^{(\ell_{s,i})} g_{t_{s,i}i} + N_i} \rp , \quad \forall 1\leq q \leq \max_{1\le\ell\le L_i}\{\hbox{order of $x_i^{\ell}$ in $\mathcal{O}_i$}\},~ i=1,2.\eeq
Now, we can formulate the sum-rate maximization problem as:
\begin{align}
\max_{\substack{\{p_i^{(\ell)}\},\mc{O}_i,~\\i=1,2}} & \sum_{i=1}^2 \sum_{\ell=1}^{L_i} r_i^{(\ell)} \label{theprobG} \\
\text{subject to: } & (\ref{eq:PowerConstraints2}), (\ref{eq:DecodConstraints}). \nn
\end{align}
%Note that we explicitly include the two decoding orders (one at each receiver) as optimization variables in \eqref{theprobG}.
Note that problem \eqref{theprobG} involves both a \emph{combinatorial optimization} of the decoding orders $\{\mc{O}_i\}$ and a \emph{non-convex optimization} of the transmit power $\{p_i^{(\ell)}\}$. As a result, it is a hard problem from an optimization point of view which has not been addressed in the literature.

Interestingly, we show that an ``indirect'' approach can effectively and fruitfully provide approximately optimal solutions to the above problem \eqref{theprobG}. Instead of directly working with the Gaussian model, we approximate the problem using the recently developed deterministic channel model \cite{ADTjour}. The approximate formulation successfully captures the key structure and intuition of the original problem, for which we give a complete analytical solution that achieves the constrained sum-capacity in all channel parameters. Next, we translate this optimal solution in the deterministic formulation back to the Gaussian formulation, and show that the resulting solution is indeed close to the optimum. This indirect approach of solving \eqref{theprobG} is outlined in Figure \ref{diagram}.

\begin{figure}[tb]
  \centering
  \includegraphics[scale = 0.7]{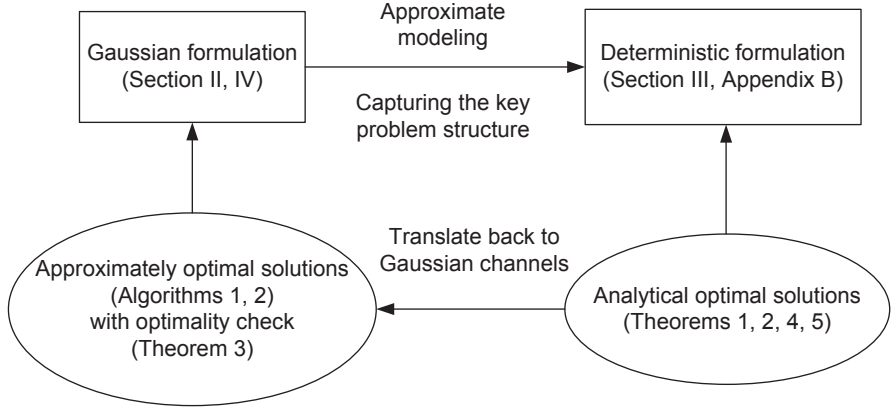}
  \caption{Our approach to solving problem \eqref{theprobG}.}
  \label{diagram}
\end{figure}

%Before moving on to the deterministic channel approximation, we end this section by showing that successive decoding does not always have the same achievability as Han-Kobayashi schemes. (This is unlike the case of multiple access channels in which the capacity region can be achieved by time-sharing of successive decoding schemes.)

Next, we provide an illustration of the following point: Although the constraints for the achievable rate region with Han-Kobayashi schemes share some similarities with those for the capacity region of multiple access channels, successive decoding in interference channels does \emph{not} always have the same achievability as Han-Kobayashi schemes, (whereas time-sharing of successive decoding schemes does achieve the capacity region of multiple access channels.)

\subsection{Successive Decoding of Gaussian Codewords vs. Gaussian Han-Kobayashi Schemes with Joint Decoding} \label{HvSsec}
We first note that Gaussian superposition coding - successive decoding is a special case of the Han-Kobayashi scheme, using the following observations.
For the $1^{st}$ user, if its message $x_1^{(\ell)} (1\le\ell\le L_1)$ is \emph{decoded} at the $2^{nd}$ receiver according to the decoding order $\mc{O}_2$, we categorize it into the \emph{common} information of the $1^{st}$ user. Otherwise, $x_1^{(\ell)}$ is treated as noise at the $2^{nd}$ receiver, i.e., it appears \emph{after} all the messages of the $2^{nd}$ user in $\mc{O}_2$, and we categorize it into the \emph{private} information of the $1^{st}$ user. The same categorization is performed for the $L_2$ messages of the $2^{nd}$ user. Note that every message of the two users is either categorized as private information or common information. Thus, every successive decoding scheme is a special case of the Han-Kobayashi scheme, and hence the capacity region with successive decoding of Gaussian codewords is included in that with Han-Kobayashi schemes.

However, the inclusion in the other direction is untrue, since Han-Kobayashi schemes allow joint decoding. In the following sections, we will give a characterization of the difference between the maximum achievable sum-rate using Gaussian successive decoding schemes and that using Gaussian Han-Kobayashi schemes. This difference appears \emph{despite} the fact that the sum-capacity of a Gaussian multiple access channel is achievable using successive decoding of Gaussian codewords. %Detailed analysis will be given in later sections. Here, we gives an illustrative example that provides some insight into this difference.
In the remainder of this section, we show an illustrative example that provides some intuition into this difference.

Suppose the $i^{th}$ user ($i=1,2$) uses \emph{two} messages: a common message $x_i^c$ and a private message $x_i^p$. %A common message is decoded at both receivers, whereas a private code is decoded at its intended receiver and treated as noise at the other receiver.
We consider a power allocation to the messages, and denote the power of $x_i^c$ and $x_i^p$ by $q_i^c$ and $q_i^p,~ (i=1,2.)$ Denote the achievable rates of $x_i^c$ and $x_i^p$ by $r_i^c$ and $r_i^p$. In a Han-Kobayashi scheme, at each receiver, the common messages and the intended private message are \emph{jointly} decoded, treating the unintended private message as noise. This gives rise to the achievable rate region with any given power allocation as follows:
\begin{align}
r_1^c + r_1^p + r_2^c \le \log(1+ \frac{q_1^c + q_1^p + g_{21}q_2^c}{g_{21}q_2^p + N_1}), ~& r_2^c + r_2^p + r_1^c \le \log(1+ \frac{q_2^c + q_2^p + g_{12}q_1^c}{g_{12}q_1^p + N_2}), \label{H1}\\
r_1^c + r_2^c \le \log(1+ \frac{q_1^c + g_{21}q_2^c}{g_{21}q_2^p + N_1}), ~& r_2^c + r_1^c \le \log(1+ \frac{q_2^c + g_{12}q_1^c}{g_{12}q_1^p + N_2}), \label{H2}\\
r_1^c + r_1^p \le \log(1+ \frac{q_1^c + q_1^p}{g_{21}q_2^p + N_1}), ~& r_2^c + r_2^p \le \log(1+ \frac{q_2^c + q_2^p}{g_{12}q_1^p + N_2}), \label{H3}\\
r_1^p + r_2^c \le \log(1+ \frac{q_1^p + g_{21}q_2^c}{g_{21}q_2^p + N_1}), ~& r_2^p + r_1^c \le \log(1+ \frac{q_2^p + g_{12}q_1^c}{g_{12}q_1^p + N_2}), \label{H4}\\
r_1^c \le \log(1+ \frac{q_1^c}{g_{21}q_2^p + N_1}), ~& r_2^c \le \log(1+ \frac{q_2^c}{g_{12}q_1^p + N_2}), \label{H5}\\
r_2^c \le \log(1+ \frac{g_{21}q_2^c}{g_{21}q_2^p + N_1}), ~& r_1^c \le \log(1+ \frac{g_{12}q_1^c}{g_{12}q_1^p + N_2}), \label{H6}\\
r_1^p \le \log(1+ \frac{q_1^p}{g_{21}q_2^p + N_1}), ~& r_2^p \le \log(1+ \frac{q_2^p}{g_{12}q_1^p + N_2}). \label{H7}
\end{align}

In a successive decoding scheme, depending on the different decoding orders applied, the achievable rate regions have different expressions. In the following, we provide and analyze the achievable rate region with the decoding orders at receiver 1 and 2 being $(x_1^c \rightarrow x_2^c \rightarrow x_1^p)$ and $(x_2^c \rightarrow x_1^c \rightarrow x_2^p)$ respectively. The intuition obtained with these decoding orders holds similarly for other decoding orders. With any given power allocation, we have
\begin{align}
r_1^c & \le \min\lp \log(1 + \frac{q_1^c}{q_1^p + g_{21}(q_2^c + q_2^p) + N_1}), \log(1+ \frac{g_{12}q_1^c}{q_2^p + g_{12}q_1^p + N_2}) \rp, \label{S1}\\
r_2^c & \le \min\lp \log(1 + \frac{q_2^c}{q_2^p + g_{12}(q_1^c + q_1^p) + N_2}), \log(1+ \frac{g_{21}q_2^c}{q_1^p + g_{21}q_2^p + N_1}) \rp, \label{S2}\\
r_1^p & \le \log(1+ \frac{q_1^p}{g_{21}q_2^p + N_1}), r_2^p \le \log(1+ \frac{q_2^p}{g_{12}q_1^p + N_2}). \label{S3}
%. \label{S4}
\end{align}
It is immediate to check that \eqref{S1} $\sim$ \eqref{S3} $\Rightarrow$ \eqref{H1} $\sim$ \eqref{H7}, but not vice versa.

To observe the difference between the constrained sum-capacity with \eqref{H1} $\sim$ \eqref{H7} and that with \eqref{S1} $\sim$ \eqref{S3}, we examine the following symmetric channel,
\begin{equation} \label{firstex}
g_{11} = g_{22} = 1, g_{12} = g_{21} = 0.17, N_1 = N_2 = 1, %, q_1^c = q_2^c = 995.83, q_1^p = q_2^p = 4.17
\end{equation}
in which we apply symmetric power allocation schemes with $q_1^c = q_2^c$ and $q_1^p = q_2^p$, and a power constraint of $\bar{p}=\bar{p}_i = q_i^p + q_i^p = 1000, i=1,2$.
\begin{RK}
Note that $\ms{SNR} = \frac{g_{11}\bar{p}}{N_i}=1000 \sim 30dB, \ms{INR} = \frac{g_{21}\bar{p}}{N_j}=170 \sim 22.5dB \Rightarrow \frac{\log\ms{INR}}{\log\ms{SNR}} \approx \frac{3}{4}$. As indicated in Figure 19 of \cite{BT08}, under this parameter setting, simply using successive decoding of Gaussian codewords can have an arbitrarily large sum-capacity loss compared to joint decoding schemes, as $\ms{SNR}\rightarrow\infty$. %In later sections, a complete characterization of will be shown as in Figure \ref{siccapafig} and Figure \ref{15_30}
\end{RK}
%To compare the maximum achievable sum-rate using a Han-Kobayashi scheme constrained by \eqref{H1} $\sim$ \eqref{H7} and that using a successive decoding scheme constrained by \eqref{S1} $\sim$ \eqref{S3},

\begin{figure}[tb]
  \centering
  \includegraphics[scale = 0.5]{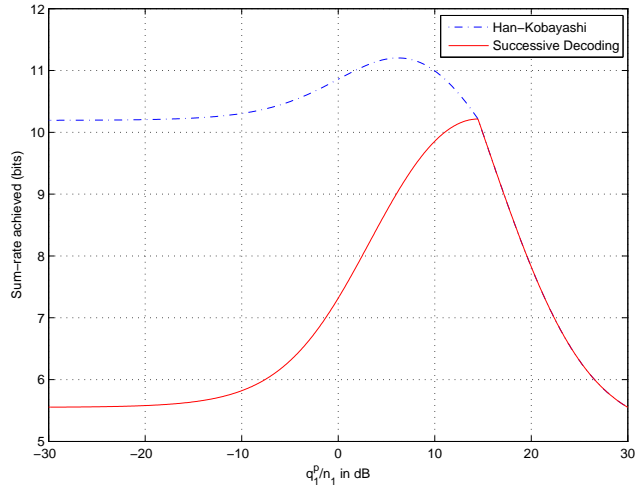}
  \caption{Illustrations of the difference between the achievable sum-rate with Han-Kobayashi schemes and that with successive decoding of Gaussian codewords.}
  \label{HKvsSC}
\end{figure}

We plot the sum-rates with the private message power $q_i^p$ sweeping from nearly zero (-30dB) to the maximum (30dB) as in Figure \ref{HKvsSC}. As observed, the difference between the two schemes is evident when the private message power $q_i^p$ is sufficiently smaller than the common message power $q_i^c$ (with $q_i^p + q_i^c = 1000$.) The intuition of why successive decoding of Gaussian codewords is not equivalent to the Han-Kobayashi schemes is best reflected in the case of $q_i^p=0$. In the above parameter setting, with $q_i^p=0$, \eqref{H1} $\sim$ \eqref{H7} translate to
\begin{align}
r_1^c + r_2^c \le \log(1+ & \frac{q_1^c + g_{21}q_2^c}{N_1}) = 10.19 ~bits, \\
r_1^c \le \log(1+ \frac{g_{12}q_1^c}{N_2}) = 7.42 & ~bits,~ r_2^c \le \log(1+ \frac{g_{21}q_2^c}{N_1}) = 7.42 ~bits,
\end{align}
whereas \eqref{S1} $\sim$ \eqref{S3} translate to
\begin{align}
r_1^c & \le \min\{ \log(1 + \frac{q_1^c}{g_{21}q_2^c + N_1}), \log(1+ \frac{g_{12}q_1^c}{N_2}) \} = \min\{2.78, 7.42\} = 2.78 ~bits,\\
r_2^c & \le \min\{ \log(1 + \frac{q_2^c}{g_{12}q_1^c + N_2}), \log(1+ \frac{g_{21}q_2^c}{N_1}) \} = \min\{2.78, 7.42\} = 2.78 ~bits.
\end{align}
As a result, the maximum achievable sum-rates with the Han-Kobayashi scheme and that with the successive decoding scheme are $10.19$ bits and $5.56$ bits respectively. Here, the key intuition is as follows: for a common message, its individual rate constraints at the two receivers in a successive decoding scheme \eqref{S1}, \eqref{S2} are tighter than those in a joint decoding scheme \eqref{H5}, \eqref{H6}. In the following sections, we will see that the constraints \eqref{S1}, \eqref{S2} lead to a non-smooth behavior of the sum-capacity using successive decoding of Gaussian codewords. Finally, we connect the results shown in Figure \ref{HKvsSC} to the results shown later in Figure \ref{15_30} of Section \ref{perf}:
\begin{RK}
In Figure \ref{HKvsSC}, the optimal symmetric power allocation for a Han-Kobayashi scheme and that for a successive decoding scheme are $q_1^p/N_1 = 6.2dB$ and $14.5dB$ respectively, leading to sum-rates of 11.2 bits and 10.2 bits. This result corresponds to the performance evaluation at $\alpha=\frac{\log(\ms{INR})}{\log(\ms{SNR})} = 0.75$ in Figure \ref{15_30}.
\end{RK}

%In the remainder of the paper, we first give a complete solution to an approximated form of \eqref{theprobG} using the deterministic interference channel model. Then, we return to the Gaussian interference channel model, and provide novel results in achievability and upper bounds of \eqref{theprobG}.

%\vspace{7pt}
\section{Sum-capacity in deterministic interference channels} \label{DCsec}

\subsection{Channel Model and Problem Formulation}
In this section, we apply the deterministic channel model \cite{ADTjour} as an approximation of the Gaussian model on the two-user interference channel. We define
\begin{align}
n_{11} \triangleq \log(\ms{SNR}_1) = \log(\frac{g_{11}\bar{p}_1}{N_1}) = \log(\tilde{g}_{11}\bar{p}_1),\\
n_{22} \triangleq \log(\ms{SNR}_2) = \log(\frac{g_{22}\bar{p}_2}{N_2}) = \log(\tilde{g}_{22}\bar{p}_2),\\
n_{12} \triangleq \log(\ms{INR}_1) = \log(\frac{g_{21}\bar{p}_2}{N_1}) = \log(\tilde{g}_{21}\bar{p}_2),\\
n_{21} \triangleq \log(\ms{INR}_2) = \log(\frac{g_{12}\bar{p}_1}{N_2}) = \log(\tilde{g}_{12}\bar{p}_1),
\end{align}
where $\tilde{g}_{ij}\triangleq g_{ij}/N_j$ are the channel gains normalized by the noise power. Without loss of generality (WLOG), we assume that $n_{11} \ge n_{22}$. We note that the logarithms used in this paper are taken to base 2. Now, $n_{ji}$ counts the bit levels of the signal sent from the $i^{th}$ transmitter that are above the noise level at the $j^{th}$ receiver. Further, we define
\begin{equation}
\delta_1 \triangleq n_{11} - n_{21} = -\log(\frac{\tilde{g}_{12}}{\tilde{g}_{11}}),~ \delta_2 \triangleq n_{22} - n_{12} = -\log(\frac{\tilde{g}_{21}}{\tilde{g}_{22}}), \label{deltadef}
\end{equation}
which represent the cross channel gains relative to the direct channel gains, in terms of the number of bit-level shifts.
To formulate the optimization problem, we consider $\{n_{ji}\}$ to be \emph{real} numbers. (As will be shown later in Remark \ref{realvsint}, with \emph{integer} bit-level channel parameters, our derivations automatically give integer bit-level optimal solutions.) %by sampling the derived sum-capacity curve, the sum-capacity with \emph{integer} bit levels can be obtained directly.

\begin{figure}[tb]
  \centering
  \includegraphics[scale = 0.6]{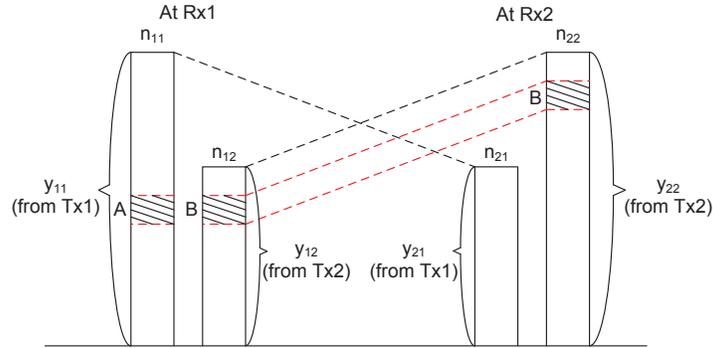}
  \caption{Two-user deterministic interference channel. Levels A and B interfere at the $1^{st}$ receiver, and cannot be fully active simultaneously.}
  \label{ICfig}
\end{figure}

In Figure \ref{ICfig}, the desired signal and the interference signal at both receivers are depicted. $y_{11}$ and $y_{12}$ are the sets of received \emph{information levels} at receiver 1 that are above the noise level, from users 1 and 2 respectively. $y_{21}$ and $y_{22}$ are the sets of received information levels at receiver 2.
%Accordingly, instead of using discrete points to represent bit levels,
A more concise representation is provided in Figure \ref{Genform}:
\begin{itemize}
\item The sets of information levels of the \emph{desired} signals at receivers 1 and 2 are represented by the continuous intervals $I_1 = [0, n_{11}]$ and $I_2 = [n_{11} - n_{22}, n_{11}]$ on two parallel lines, where the leftmost points correspond to the most significant (i.e., highest) information levels, and the points at $n_{11}$ correspond to the positions of the noise levels at both receivers.
\item The positions of the information levels of the \emph{interfering} signals are indicated by the dashed lines crossing between the two parallel lines.
\end{itemize}

Note that an information level (or simply termed \emph{``level''}) is a real \emph{point} on a line, and the measure of a set of levels (e.g. the length of an interval) equals the amount of information that this set can carry. The design variables %in the deterministic channel
are \emph{whether each level of a user's received desired signal carries information for this user}, characterized by the following definition:
\begin{define}
$f_i(x)$ is the indicator function on whether the levels inside $I_i$ carry information for the $i^{th}$ user.
\begin{equation} \label{indic}
f_i(x) = \left\{
\begin{array}{lll}
1, & \text{ if } x\in I_i, \text{ and level } x \text{ carries information for the $i^{th}$ user,}\\
0, & \text{ otherwise.}
\end{array}
\right.
~ (i = 1,2.)
\end{equation}
\end{define}
As a result, the rates of the two users are $R_1 = \int_0^{n_{11}} f_1(x) \mr{d}x, R_2 = \int_0^{n_{11}} f_2(x) \mr{d}x$. For an information level $x$ s.t. $f_i(x)=1$, we call it an \emph{active} level for the $i^{th}$ user, and otherwise an \emph{inactive} level.

The constraints from superposition of Gaussian codewords with successive decoding \eqref{S1} $\sim$ \eqref{S3} translate to the following \emph{Complementarity Conditions} in the deterministic formulation.
%\begin{lem} [Complementarity Conditions] \label{lemCC}
\begin{align}
& f_1(x)f_2(x+\delta_1) = 0, \forall -\infty < x < \infty \label{cc10},\\
& f_2(x)f_1(x+\delta_2) = 0, \forall -\infty < x < \infty \label{cc20},
\end{align}
where $\delta_1$ and $\delta_2$ are defined in \eqref{deltadef}.
The interpretation of \eqref{cc10} and \eqref{cc20} are as follows: for any two levels each from one of the two users, if they interfere with each other at any of the two receivers, they cannot be simultaneously active. For example, in Figure \ref{ICfig}, information levels $A$ from the $1^{st}$ user and $B$ from the $2^{nd}$ user interfere at the $1^{st}$ receiver, and hence cannot be fully active simultaneously. These complementarity conditions have also been characterized using a conflict graph model in \cite{Shaoinfo}.

\begin{figure}[tb]
  \centering
  \includegraphics[scale = 0.7]{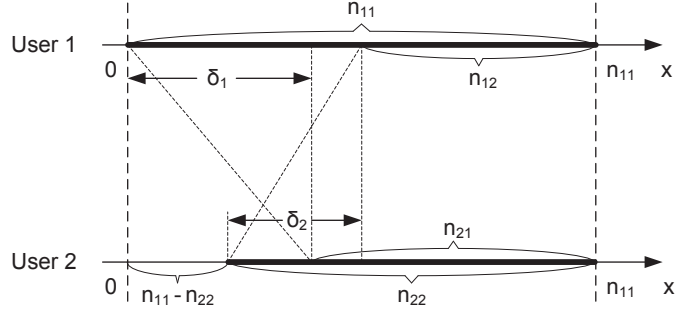}
  \caption{Interval representation of the two-user deterministic interference channel.}
  \label{Genform}
\end{figure}

%The relation between the superposition codes and $f_i(x), (i=1,2)$ is given in the following remark.
\begin{RK} \label{numcodes}
For any given function $f_i(x), x\in I_i$, every disjoint segment within $I_i$ with $f_i(x) = 1$ on it corresponds to a distinct message. Adjacent segments that can be so combined as a super-segment having $f_i(x) = 1$ on it, are viewed as \emph{one} segment, i.e., the combined super-segment. Thus, for two segments $s_1=[a,b]\in I_i$ and $s_2=[c,d]\in I_i, (b<c,)$ satisfying $f_i(x) = 1, \forall x\in s_1\cup s_2$, if $\exists x_0\in(b,c), f(x_0)=0$, then $s_1,s_2$ separated by the point $x_0$ have to correspond to two distinct messages.
\end{RK}

Finally, we note that
\begin{align}
&\eqref{cc10} \Leftrightarrow f_2(x)f_1(x - \delta_1) = 0, \forall -\infty < x < \infty, \nn\\
\text{and}~~& \eqref{cc20} \Leftrightarrow f_1(x)f_2(x-\delta_2) = 0, \forall -\infty < x < \infty.\nn
\end{align}
Thus, we have the following result:
\begin{lem} \label{equivcc}
The parameter settings $\left\{
\begin{array}{l}
\delta_1 = a\\
\delta_2 = b
\end{array}
\right.$
and $\left\{
\begin{array}{l}
\delta_1 = -b\\
\delta_2 = -a
\end{array}
\right.$ correspond to the same set of complementarity conditions.
\end{lem}

We consider the problem of maximizing the sum-rate $R^{sum} \triangleq R_1 + R_2$ of the two users employing successive decoding, formulated as the following continuous support (infinite dimensional) optimization problem:
\begin{align}
\max_{f_1(x),f_2(x)} & (R^{sum} = ) {\int_0^{n_{11}} f_1(x) + f_2(x) \mr{d}x} \label{theprob0} \\
\text{ subject to } & \eqref{indic}, \eqref{cc10}, \eqref{cc20}. \nn
\end{align}
%We will see that the maximum achievable sum-rate ${R^{sum}}^*$ is always achievable with $R_1 = R_2$. Thus, we also use the maximum achievable \emph{symmetric} rate, denoted by $R$, as an equivalent performance measure.
Problem \eqref{theprob0} does not include upper bounds on the number of messages $L_1,L_2$. Such upper bounds can be added based on Remark \ref{numcodes}. We will analyze the cases without and with upper bounds on the number of messages. We first derive the constrained sum-capacity in \emph{symmetric} interference channels in the remainder of this section. Results are then generalized using similar approaches to \emph{general (asymmetric)} interference channels in Appendix \ref{asymDC}.

%\vspace{5pt}
\subsection{Symmetric Interference Channels}
In this section, we consider the case where $n_{11} = n_{22}, n_{12} = n_{21}$. Define $\alpha \triangleq \frac{n_{12}}{n_{11}}, \beta \triangleq 1-\alpha$. WLOG, we normalize the amount of information levels by $n_{11}$, and consider $n_{11} = n_{22} = 1$, and $n_{12} = n_{21} = \alpha$. %(Note that this normalization is, however, \emph{not} WLOG in the Gaussian channel.)
Note that in symmetric channels, $\beta = \delta_1 = \delta_2$. %However, we still introduce $\beta$ as a new notation, because $\beta$ will represent something different from $\delta_1$ (or $\delta_2$) in \emph{asymmetric} channels (Appendix \ref{asymwo}.)

Now, \eqref{cc10} \eqref{cc20} becomes
\begin{align}
& f_1(x)f_2(x+\beta) = 0, \forall -\infty < x < \infty, \label{cc1}\\
& f_2(x)f_1(x+\beta) = 0, \forall -\infty < x < \infty. \label{cc2}
\end{align}
Problem \eqref{theprob0} becomes
\begin{align}
\max_{f_1(x),f_2(x)} & (R^{sum} = ) {\int_0^1 f_1(x) + f_2(x) \mr{d}x} \label{theprob} \\
\text{ subject to } & \eqref{indic}, \eqref{cc1}, \eqref{cc2}. \nn
\end{align}

From Lemma \ref{equivcc}, \emph{it is sufficient to only consider the case with $\beta\ge 0$, i.e. $\alpha \le 1$.}

%\vspace{5pt}
We next derive the constrained sum-capacity using successive decoding for $\alpha\in [0,1]$, first without upper bounds on the number of messages, then with upper bounds. We will see that in symmetric channels, the constrained sum-capacity ${R^{sum}}^*$ is achievable with $R_1 = R_2$. Thus, we also use the maximum achievable \emph{symmetric} rate, denoted by $R(\alpha)$ as a function of $\alpha$, as an equivalent performance measure. $R(\alpha)$ is thus one half of the optimal value of \eqref{theprob}.

\vspace{8pt}
\subsubsection{Symmetric Capacity without Constraint on the Number of Messages}~
%\vspace{5pt}

\begin{figure}[tb]
  \centering
  \includegraphics[scale = 0.7]{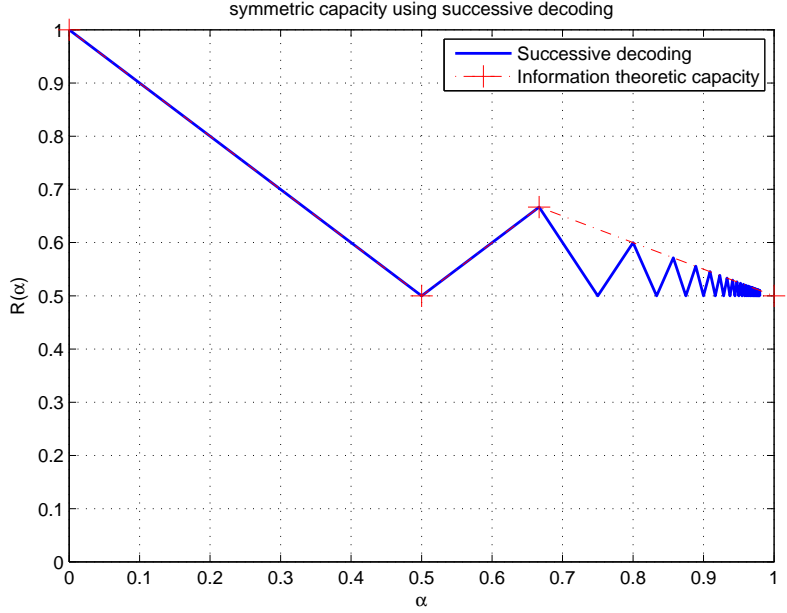}
  \caption{The symmetric capacity with successive decoding in symmetric deterministic interference channels.}
  \label{siccapafig}
\end{figure}

\begin{thm} \label{siccapa}
The maximum achievable symmetric rate using successive decoding, (i.e., having constraints \eqref{cc1}, \eqref{cc2}), $R(\alpha)~(\alpha \in [0,1])$, is characterized by
\begin{itemize}
\item $R(\alpha) = 1-\frac{\alpha}{2}$, when $\alpha = \frac{2n}{2n+1}, n = 0,1,2,\ldots$.
\item $R(\alpha) = \frac{1}{2}$, when $\alpha = \frac{2n-1}{2n}, n = 1,2,3,\ldots$.
\item In every interval $[\frac{2n}{2n+1},\frac{2n+1}{2n+2}], n = 0,1,2,\ldots$, $R(\alpha)$ is a decreasing linear function.
\item In every interval $[\frac{2n-1}{2n},\frac{2n}{2n+1}], n = 1,2,3,\ldots$, $R(\alpha)$ is an increasing linear function.
\item $R(1) = \frac{1}{2}$.
\end{itemize}
\end{thm}
\begin{RK} We plot $R(\alpha)$ in Figure \ref{siccapafig}, compared with the information theoretic capacity \cite{BT08}. \end{RK}

The key ideas in deriving the constrained sum-capacity are to \emph{decompose} the effects of the complementarity conditions, such that the resulting sub-problems become easier to solve.

\begin{proof}[Proof of Theorem \ref{siccapa}]

i) When $\frac{2n-1}{2n}< \alpha \le \frac{2n}{2n+1}, n = 1,2,3,\ldots, \frac{1}{2n+1} \le \beta < \frac{1}{2n}$. We divide the interval $[0,1]$ into $2n+1$ segments $\{s_1,\ldots,s_{2n+1}\}$, where the first $2n$ segments have length $\beta$, and the last segment has length $1-2n\beta\in(0,\frac{1}{2n+1}]$ (cf. Figure \ref{Inc}.) With these, the complementarity conditions \eqref{cc1} \eqref{cc2} are equivalent to the following:
\begin{align}
~&
\left\{ %\eqref{cc1} \Leftrightarrow
\begin{array}{rl}
& \forall x\in s_1 (\Leftrightarrow x+\beta\in s_2), f_1(x)f_2(x+\beta) = 0, \\
& \forall x\in s_2 (\Leftrightarrow x+\beta\in s_3), f_2(x)f_1(x+\beta) = 0, \\
& ~~~~~~~~~~~~~~~~~~~~~~~\cdots\\
%f_2(x)f_1(x+\beta) = 0, & \forall x\in b_{2n-2} (\Leftrightarrow x+\beta\in a_{2n-1}) \\
& \forall x\in s_{2n-1} (\Leftrightarrow x+\beta\in s_{2n}), f_1(x)f_2(x+\beta) = 0,
\end{array}
\right.\label{group1}\\
\text{and}~~ & ~~~~ \forall x+\beta\in s_{2n+1}, ~ f_2(x)f_1(x+\beta) = 0, \label{group1e}
\end{align}
(Relations \eqref{group1} and \eqref{group1e} correspond to the shaded strips in Figure \ref{Inc}.)

Similarly,
\begin{align}
~&
\left\{
\begin{array}{rl}
& \forall x\in s_1 (\Leftrightarrow x+\beta\in s_2), f_2(x)f_1(x+\beta) = 0, \\
& \forall x\in s_2 (\Leftrightarrow x+\beta\in s_3), f_1(x)f_2(x+\beta) = 0, \\
& ~~~~~~~~~~~~~~~~~~~~~~~\cdots\\
%f_1(x)f_2(x+\beta) = 0, & \forall x\in a_{2n-2} (\Leftrightarrow x+\beta\in b_{2n-1}) \\
& \forall x\in s_{2n-1} (\Leftrightarrow x+\beta\in s_{2n}), f_2(x)f_1(x+\beta) = 0,
\end{array}
\right.\label{group2}\\
\text{and}~~ & ~~~~ \forall x+\beta\in s_{2n+1}, ~ f_1(x)f_2(x+\beta) = 0. \label{group2e}
\end{align}

\begin{figure}[tb]
  \centering
  \includegraphics[scale = 0.6]{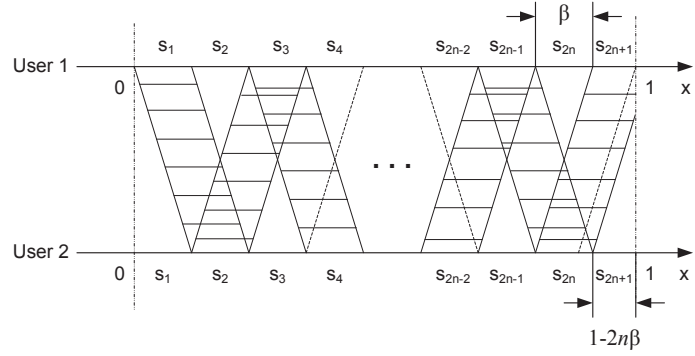}
  \caption{Segmentation of the information levels, $\frac{2n-1}{2n}< \alpha \le \frac{2n}{2n+1}$.}
  \label{Inc}
\end{figure}

We partition the set of all segments into two groups:
\begin{equation*}\mc{G}_1 = s_1\cup s_3\cup\ldots\cup s_{2n+1} \text{ and } \mc{G}_2 = s_2\cup s_4\cup \ldots\cup s_{2n}.\end{equation*}
Note that
\begin{itemize}
\item \eqref{group1} and \eqref{group1e} are constraints on $f_1(x)$ with support in $\mc{G}_1$, and on $f_2(x)$ with support in $\mc{G}_2$.
\item \eqref{group2} and \eqref{group2e} are constraints on $f_1(x)$ with support in $\mc{G}_2$, and on $f_2(x)$ with support in $\mc{G}_1$.
\end{itemize}
Consequently, instead of viewing the (infinite number of) optimization variables as $f_1(x)|_{[0,1]}$ and $f_2(x)|_{[0,1]}$, it is more convenient to view them as
\begin{equation}
C_1\triangleq\{f_1(x)|_{\mc{G}_1}, f_2(x)|_{\mc{G}_2}\} \text{ and } C_2\triangleq\{f_1(x)|_{\mc{G}_2}, f_2(x)|_{\mc{G}_1}\}, \label{c1c2}
\end{equation}
because there is \emph{no constraint between $C_1$ and $C_2$} from the complementarity conditions. In other words, $C_1$ and $C_2$ can be optimized \emph{independently} of each other. Define
\begin{align}
R_{C_1}^{sum} \triangleq \int_{\mc{G}_1} f_1(x) \mr{d}x + \int_{\mc{G}_2} f_2(x) \mr{d}x, \nn\\
R_{C_2}^{sum} \triangleq \int_{\mc{G}_2} f_1(x) \mr{d}x + \int_{\mc{G}_1} f_2(x) \mr{d}x. \nn
\end{align}
Clearly, $R^{sum} = R_{C_1}^{sum} + R_{C_2}^{sum}$. Hence \eqref{theprob} can be solved by separately solving the following two sub-problems:
\begin{align}
\max_{f_1(x)|_{\mc{G}_1}, f_2(x)|_{\mc{G}_2}} & (R_{C_1}^{sum}=)\int_{\mc{G}_1} f_1(x) \mr{d}x + \int_{\mc{G}_2} f_2(x) \mr{d}x \label{prob1} \\
\text{ subject to } & \eqref{indic}, \eqref{group1}, \eqref{group1e}, \nn
\end{align}
and
\begin{align}
\max_{f_1(x)|_{\mc{G}_2}, f_2(x)|_{\mc{G}_1}} & (R_{C_2}^{sum}=)\int_{\mc{G}_2} f_1(x) \mr{d}x + \int_{\mc{G}_1} f_2(x) \mr{d}x \label{prob2} \\
\text{ subject to } & \eqref{indic}, \eqref{group2}, \eqref{group2e}. \nn
\end{align}

We now prove that the optimal value of \eqref{prob1} is ${R_{C_1}^{sum}}^*=1-n\beta$:
\begin{itemize}
\item (Achievability:) $1-n\beta$ is achievable with $f_1(x) = 1,\forall x\in\mc{G}_1$, and $f_2(x) = 0,\forall x\in\mc{G}_2$.
\item (Converse:) $\eqref{group1} \Rightarrow \forall i\in\{1,2,\ldots,n\},~\int_{s_{2i-1}} f_1(x) \mr{d}x + \int_{s_{2i}} f_2(x) \mr{d}x \le \beta$
\begin{align}
\Rightarrow \int_{\mc{G}_1} f_1(x) \mr{d}x + \int_{\mc{G}_2} f_2(x) \mr{d}x = & \sum_{i = 1}^n \big(\int_{s_{2i-1}} f_1(x) \mr{d}x + \int_{s_{2i}} f_2(x) \mr{d}x \big) + \int_{s_{2i+1}} f_1(x) \mr{d}x \nn\\
\le & \beta\cdot n + (1-2n\beta) = 1-n\beta.
\end{align}
\end{itemize}
By symmetry, the solution of \eqref{prob2} can be obtained similarly, and the optimal value is ${R_{C_2}^{sum}}^* = 1-n\beta$ as well. Therefore, the optimal value of \eqref{theprob} is ${R^{sum}}^* = 2(1-n\beta)$.

As the above maximum achievable scheme is symmetric, i.e.,
\begin{equation} \label{achieve}
f_1(x) = f_2(x) = \left\{
\begin{array}{rl}
1, & \forall x\in \mc{G}_1 \\
0, & \forall x\in \mc{G}_2 \\
\end{array}
\right. ,
\end{equation}
the symmetric capacity is
\begin{equation}
R(\alpha) = 1-n\beta = n\alpha + 1-n.
\end{equation}
Clearly, $R(\alpha)$ is an \emph{increasing linear} function of $\alpha$ in every interval $(\frac{2n-1}{2n}, \frac{2n}{2n+1}], n = 1,2,3,\ldots$. It can be verified that $R(\alpha)|_{\frac{2n-1}{2n}} = \frac{1}{2}$, and $R(\alpha)|_{\frac{2n}{2n+1}} = 1 - \frac{\alpha}{2}$.

\vspace{9pt}
ii) When $\frac{2n}{2n+1}< \alpha \le \frac{2n+1}{2n+2}, n = 0,1,2,\ldots, \frac{1}{2n+2} \le \beta < \frac{1}{2n+1}$. Similarly to i), we divide the interval $[0,1]$ into $2n+2$ segments $\{s_1,\ldots,s_{2n+2}\}$, where the first $2n+1$ segments have length $\beta$, and the last segment has length $1-(2n+1)\beta\in(0,\frac{1}{2n+2}]$ (cf. Figure \ref{Dec}). Then, the complementarity conditions \eqref{cc1}, \eqref{cc2} are equivalent to the following:
\begin{align}
~& \eqref{group1}, \eqref{group1e} \text{ and } f_1(x)f_2(x+\beta) = 0, \forall x+\beta\in s_{2n+2}, \label{group1ee}\\
\text{and}~~ & \eqref{group2}, \eqref{group2e} \text{ and } f_2(x)f_1(x+\beta) = 0, \forall x+\beta\in s_{2n+2}. \label{group2ee}
\end{align}

\begin{figure}[tb]
  \centering
  \includegraphics[scale = 0.6]{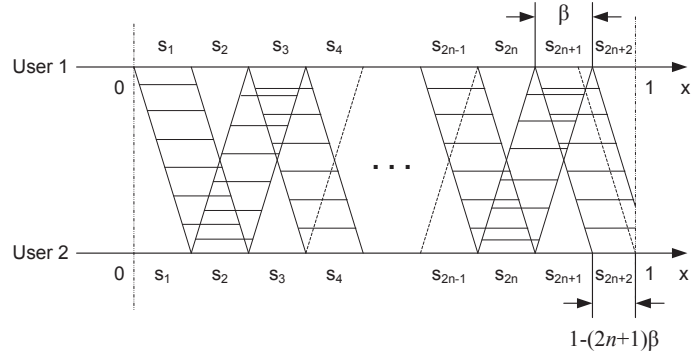}
  \caption{Segmentation of the information levels, $\frac{2n}{2n+1}< \alpha \le \frac{2n+1}{2n+2}$.}
  \label{Dec}
\end{figure}

Similarly to i), with $\mc{G}_1 = s_1\cup s_3\cup \ldots\cup s_{2n+1}$ and $\mc{G}_2 = s_2\cup s_4\cup \ldots\cup s_{2n+2}$, \eqref{theprob} can be solved by separately solving the following two sub-problems:
\begin{align}
\max_{f_1(x)|_{\mc{G}_1}, f_2(x)|_{\mc{G}_2}} & (R_{C_1}^{sum}=)\int_{\mc{G}_1} f_1(x) \mr{d}x + \int_{\mc{G}_2} f_2(x) \mr{d}x \label{prob11} \\
\text{ subject to } & \eqref{indic}, \eqref{group1}, \eqref{group1e}, \eqref{group1ee}, \nn
\end{align}
and
\begin{align}
\max_{f_1(x)|_{\mc{G}_2}, f_2(x)|_{\mc{G}_1}} & (R_{C_2}^{sum}=)\int_{\mc{G}_2} f_1(x) \mr{d}x + \int_{\mc{G}_1} f_2(x) \mr{d}x \label{prob22} \\
\text{ subject to } & \eqref{indic}, \eqref{group2}, \eqref{group2e}, \eqref{group2ee}. \nn
\end{align}

We now prove that the optimal value of \eqref{prob11} is $(n+1)\beta$:
\begin{itemize}
\item (Achievability:) $(n+1)\beta$ is achievable with $f_1(x) = 1,\forall x\in\mc{G}_1$, and $f_2(x) = 0,\forall x\in\mc{G}_2$.
\item (Converse:) $\eqref{group1}, \eqref{group1e}, \eqref{group1ee} \Rightarrow \forall i\in\{1,2,\ldots,n+1\},~\int_{s_{2i-1}} f_1(x) \mr{d}x + \int_{s_{2i}} f_2(x) \mr{d}x \le \beta$
\begin{align}
\Rightarrow \int_{\mc{G}_1} f_1(x) \mr{d}x + \int_{\mc{G}_2} f_2(x) \mr{d}x = & \sum_{i = 1}^{n+1} \big(\int_{s_{2i-1}} f_1(x) \mr{d}x + \int_{s_{2i}} f_2(x) \mr{d}x \big) \nn\\
\le & (n+1)\beta.
\end{align}
\end{itemize}
By symmetry, the solution of \eqref{prob22} can be obtained similarly. Thus, the optimal value of \eqref{theprob} is $2(n+1)\beta$. The maximum achievable scheme is also characterized by \eqref{achieve}, and the symmetric rate is
\begin{equation}
R(\alpha) = (n+1)\beta = -(n+1)\alpha + n+1.
\end{equation}
Clearly, $R(\alpha)$ is a \emph{decreasing linear} function of $\alpha$ in every interval $(\frac{2n}{2n+1}, \frac{2n+1}{2n+2}], n = 0,1,2,\ldots$. It can be verified that $R(\alpha)|_{\frac{2n}{2n+1}} = 1 - \frac{\alpha}{2}$, and $R(\alpha)|_{\frac{2n+1}{2n+2}} = \frac{1}{2}$.

\vspace{9pt}
iii) It is clear that $R(0) = 1$, which is achievable with $f_1(x) = f_2(x) = 1, \forall x\in (0,1)$, and $R(1) = \frac{1}{2}$, which is achievable by time sharing $\left\{\begin{array}{rl} f_1(x) = 1, & x\in[0,1]\\ f_2(x) = 0, & x\in[0,1] \end{array}\right.$ and $\left\{\begin{array}{rl} f_1(x) = 0, & x\in[0,1]\\ f_2(x) = 1, & x\in[0,1] \end{array}\right.$.
\end{proof}

\vspace{10pt}
We summarize the optimal scheme that achieves the constrained symmetric capacity as follows:
\begin{cor} \label{optachsym}
When $\alpha\in(0,1)$, the constrained symmetric capacity is achievable with
\beq \label{achieve0}
f_1(x) = f_2(x) = \left\{
\begin{array}{rl}
1, & \forall x\in \mc{G}_1 \\
0, & \forall x\in \mc{G}_2 \\
\end{array}
\right.,
\eeq
where $\mc{G}_1 = \bigcup_{i=1,2,\ldots}s_{2i-1}$ and $\mc{G}_2 = \bigcup_{i=1,2,\ldots}s_{2i}$.
\end{cor}

In the special cases when $\alpha = \frac{2n-1}{2n}, (n = 1,2,3,\ldots,)$ and $\alpha = 1$, the constrained symmetric capacity drops to $\frac{1}{2}$ which is also achievable by time sharing $\left\{\begin{array}{rl} f_1(x) = 1, & x\in[0,1]\\ f_2(x) = 0, & x\in[0,1] \end{array}\right.$ and $\left\{\begin{array}{rl} f_1(x) = 0, & x\in[0,1]\\ f_2(x) = 1, & x\in[0,1] \end{array}\right.$.

We observe that the \emph{numbers of messages used} by the two users --- $L_1, L_2$ --- in the above optimal schemes are as follows:
\begin{cor} \label{neceL}~

\begin{itemize}
\item when $\alpha\in(\frac{2n-1}{2n}, \frac{2n+1}{2n+2}), (n = 1,2,3,\ldots)$, $L_1 = L_2 = n+1$;
\item when $\alpha\in[0, \frac{1}{2}]$, $\alpha = \frac{2n-1}{2n}, (n = 1,2,3,\ldots)$, or $\alpha=1$, $L_1 = L_2 = 1$.
\end{itemize}
\end{cor}

\begin{RK} \label{realvsint}
%We have discussed the problem where $\alpha$ can take any \emph{real} value.
In the original formulation of the deterministic channel model \cite{ADTjour}, $\{n_{ij}\}$ are considered to be \emph{integers}, and the achievable scheme must also have integer bit-levels. In this case, $\alpha = \frac{n_{12}}{n_{11}}$ is a \emph{rational} number. As a result, the optimal scheme \eqref{achieve0} will consist of active segments $\mc{G}_1$ that have rational boundaries with the same denominator $n_{11}$. %Thus, after multiplying by $n_{11}$ (i.e., removing the normalization,) these active segments (messages) all correspond to integer bit levels automatically.
This indeed corresponds to an integer bit-level solution.
\end{RK}
%We will see that the intuitions in deriving the optimal scheme for the symmetric channels also extend to asymmetric channels in Section \ref{asymDC}.

From Theorem \ref{siccapa} (cf. Figure \ref{siccapafig}), it is interesting to see that the constrained symmetric capacity oscillates as a function of $\alpha$ between the information theoretic capacity and the baseline of $1/2$. This phenomenon is a consequence of the complementarity conditions. In Section \ref{disc}, we further discuss the connections of this result to other coding-decoding constraints.

\vspace{8pt}
\subsubsection{The Case with a Limited Number of Messages}~ \label{limsym}
%\vspace{5pt}

In this subsection, we find the maximum achievable sum/symmetric rate using successive decoding when there are constraints on the maximum number of messages for the two users respectively. Clearly, the constrained symmetric capacity achieved with $\alpha\in[0,1]$ will be lower than $R(\alpha)$. We start with the following two lemmas, whose proofs are relegated to Appendix \ref{prflemlim}:

\begin{lem} \label{noweak}
If there exists a segment with an even index $s_{2i}~(i\ge1)$ and $s_{2i}$ does not end at 1, such that
\begin{equation*}
f_1(x) = 1,\forall x\in s_{2i}, \text{ \emph{or} } f_2(x) = 1,\forall x\in s_{2i},
\end{equation*}
(with $f_i(x)$ defined as in \eqref{indic},) then $R^{sum} \le 1$.
\end{lem}
\begin{lem} \label{ifnostrong}
If there exists a segment with an odd index $s_{2i-1}~(i\ge1)$, such that
\begin{equation*}
f_1(x) = 0,\forall x\in s_{2i-1}, \text{ \emph{or} } f_2(x) = 0,\forall x\in s_{2i-1},
\end{equation*}
then $R^{sum} \le 1$.
\end{lem}

Recall that the optimal scheme \eqref{achieve0} requires that, for both users, \emph{all} segments in $\mc{G}_2$ are fully inactive, and \emph{all} segments in $\mc{G}_1$ are fully active. The above two lemmas show the cost of violating \eqref{achieve0}: if one of the segments in $\mc{G}_2$ becomes fully active for either user (cf. Lemma \ref{noweak}), or one of the segments in $\mc{G}_1$ becomes fully inactive for either user (cf. Lemma \ref{ifnostrong}), the resulting sum-rate cannot be greater than 1.
We now establish the following theorem:
\begin{thm} \label{lessL}
Denote by $L_i (i=1,2)$ the number of messages used by the $i^{th}$ user. When $\alpha\in(\frac{2n-1}{2n}, \frac{2n+1}{2n+2}), (n = 1,2,\ldots,)$ if $L_1\le n$ \emph{or} $L_2\le n$, the maximum achievable sum-rate is 1.
\end{thm}

%\vspace{5pt}
\begin{proof}
WLOG, assume that there is a constraint of $L_1\le n$.

i) First, the sum-rate of 1 is always achievable with
\begin{equation*}
f_1(x) = 1, f_2(x) = 0, \forall x\in[0,1].
\end{equation*}

ii) If there exists $s_{2i}, 1\le i\le n$, such that \emph{either} $f_1(x) = 1,\forall x\in s_{2i}$, \emph{or} $f_2(x) = 1,\forall x\in s_{2i}$, then from Lemma \ref{noweak}, the achieved sum-rate is no greater than 1.

iii) If \emph{for all} $s_{2i}, 1\le i\le n$, there exists $x_{i}$ in the interior of $s_{2i}$ such that $f_1(x_{i})= 0$:

Note that $x_{i}$ \emph{separates} the two segments $s_{2i-1}, s_{2i+1}$ for the $1^{st}$ user. From Remark \ref{numcodes}, $s_{2i-1}$ and $s_{2i+1}$ have to be \emph{two distinct messages} provided that both of them are (at least partly) active for the $1^{st}$ user. On the other hand, there are $n+1$ such segments $\mc{G}_1 = \{s_1,s_3,\ldots,s_{2n+1}\}$ (cf. Figures \ref{Inc} and \ref{Dec}), whereas the number of messages of the $1^{st}$ user is upper bounded by $L_1\le n$. Consequently, $\exists 1\le i_1\le n+1$, such that $f_1(x) = 0, \forall x\in s_{2i_1-1}$. In other words, there must be a segment in $\mc{G}_1$ that is fully inactive for the $1^{st}$ user. By Lemma \ref{ifnostrong}, in this case, the achieved sum-rate is no greater than 1.
\end{proof}

Comparing Theorem \ref{lessL} with Corollary \ref{neceL}, we conclude that if the number of messages used for \emph{either} of the two users is fewer than the number used in the optimal scheme \eqref{achieve0} (as in Corollary \ref{neceL}), the maximum achievable symmetric rate drops to $\frac{1}{2}$.
This is illustrated in Figure \ref{L2capa} with $L_1\le 2$ (or $L_2\le 2$), and in Figure \ref{L3capa} with $L_1\le 3$ (or $L_2\le 3$).

\begin{figure}[tb!]
  \centering
  \subfigure[Maximum achievable symmetric rate with $L_1\le 2$.]{
  \includegraphics[scale = 0.4]{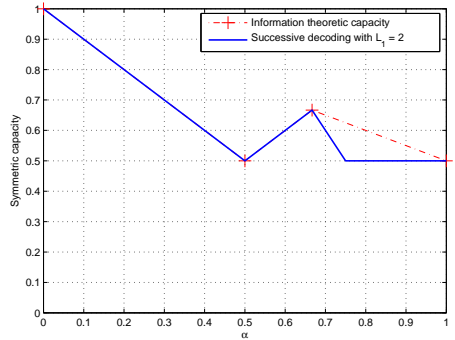}
  \label{L2capa}
  }
  \subfigure[Maximum achievable symmetric rate with $L_1\le 3$.]{
  \includegraphics[scale = 0.4]{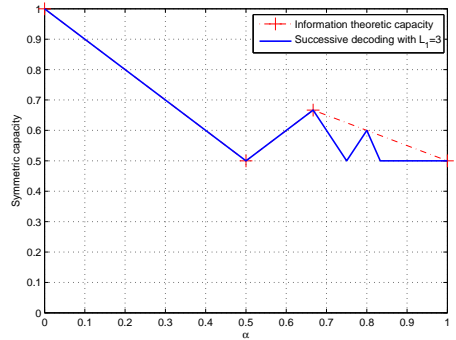}
  \label{L3capa}
  }
  \caption{The symmetric capacity with a limited number of messages.}
  \label{suball}
\end{figure}

\vspace{5pt}
Complete solutions (without and with constraints on the number of messages) in \emph{asymmetric} channels follow similar ideas, albeit more tediously. Detailed discussions are relegated to Appendix \ref{asymDC}.

\section{Approximate Sum-capacity for Successive Decoding in Gaussian Interference Channels} \label{GCsec} %Achievability and upper bounds on the sum-rate
In this section, we turn our focus back to the two-user Gaussian interference channel, and consider the sum-rate maximization problem \eqref{theprobG}. Based on the relation between the deterministic channel model and the Gaussian channel model, we \emph{translate} the optimal solution of the deterministic channel into the Gaussian channel. We then derive upper bounds on the optimal value of \eqref{theprobG}, and evaluate the achievability of our translation against these upper bounds.

\subsection{Achievable Sum-rate Motivated by the Optimal Scheme in the Deterministic Channel}
As the deterministic channel model can be viewed as an approximation to the Gaussian channel model, optimal schemes of the former suggest approximately optimal schemes of the latter. In this subsection, we show the translation of the optimal scheme of the deterministic channel to that of the Gaussian channel. We show in detail \emph{two forms} (simple and fine) of the translation for symmetric interference channels:
\begin{equation*}
g_{11} = g_{22}, g_{12} = g_{21}, N_1 = N_2, \bar{p}_1 = \bar{p}_2 = \bar{p}.
\end{equation*}
The translation for asymmetric channels can be derived similarly, albeit more tediously.

%\vspace{7pt}
\subsubsection{A simple translation of power allocation for the messages}~

\begin{figure}[b!]
  \centering
  \includegraphics[scale = 0.7]{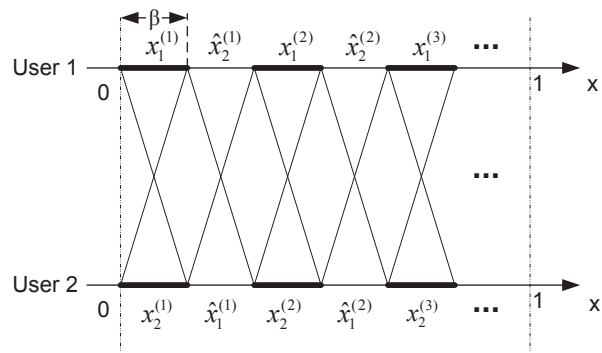}
  \caption{The optimal scheme in the symmetric deterministic interference channel.}
  \label{approx1}
\end{figure}

Recall the optimal scheme for symmetric deterministic interference channels (Corollary \ref{optachsym},) as plotted in Figure \ref{approx1}. $x_i^{(\ell)}, \ell = 1,\ldots,L$ represent the segments (or \emph{messages} as translated to the Gaussian channel) that are active for the $i^{th}$ user. Recall that % where $a_1,a_2,\ldots$ and $b_1,b_2,\ldots$ represent the codes of the $1^{st}$ and the $2^{nd}$ user respectively.
\beq
- \beta = - (1-\alpha) = n_{21} - n_{11} = \log(\frac{g_{12}}{g_{11}}).
\eeq
Thus, a shift of $\beta$ to the right (i.e. lower information levels) in the deterministic channel approximately corresponds to a power scaling factor of $\frac{g_{12}}{g_{11}}$ in the Gaussian channel. Accordingly, a simple translation of the symmetric optimal scheme (cf. Figure \ref{approx1}) into the Gaussian channel is given as follows:

%\vspace{7pt}
%\begin{alg}~

%\begin{itemize}
%\item{Step 1: Determine the number of codes $L_1 = L_2 = L$ for each user as the same number used in the optimal deterministic channel scheme.}
%\item{Step 2: Let $\frac{p^{(2)}}{p^{(1)}}=\frac{p^{(3)}}{p^{(2)}} = \ldots = \frac{p^{(L)}}{p^{(L-1)}} = \big(\frac{g_{12}}{g_{11}}\big)^2$, and normalize the power by $\sum_{\ell = 1}^{L}p^{(\ell)} = \bar{p}$}.
%\end{itemize}
%\end{alg}

\begin{tabular}{l} \label{alg1}
\vspace{-14pt}\\

\emph{Algorithm 1: A simple translation by direct power scaling.}\\
\vspace{-15pt}\\
\hline
Step 1: Determine the number of messages $L_1 = L_2 = L$ for each user as the same number used in the \\ $~~~~~~~~~$ optimal deterministic channel scheme.\\

Step 2: Let $\frac{p^{(2)}}{p^{(1)}}=\frac{p^{(3)}}{p^{(2)}} = \ldots = \frac{p^{(L)}}{p^{(L-1)}} = \big(\frac{g_{12}}{g_{11}}\big)^2$, and normalize the power by $\sum_{\ell = 1}^{L}p^{(\ell)} = \bar{p}$.\\
\vspace{-15pt}\\
\hline
\\
\end{tabular}

\vspace{7pt}
\subsubsection{A finer translation of power allocation for the messages}~

In this part, for notational simplicity, we assume WLOG that the noise power $N_{1} = N_{2} = 1$ and $g_{11} = 1$. We consider the case where the cross channel gain is no greater than the direct channel gain% (weak interference channel)
: $0\le g_{12}\le g_{11}$.

In the optimal deterministic scheme, the key property that ensures optimality is the following:
\begin{cor} \label{eqrate}
A message $x_i^{(\ell)}$ that is decoded at \emph{both} receivers is subject to the \emph{same} achievable rate constraint at both receivers.
\end{cor}

For example, In the optimal deterministic scheme (cf. Figure \ref{approx1}), message $x_1^{(1)}$ is subject to an achievable rate constraint of $|x_1^{(1)}|$ at the $1^{st}$ receiver, and that of $|\hat{x}_1^{(1)}|$ at the $2^{nd}$ receiver, with $|x_1^{(1)}|=|\hat{x}_1^{(1)}| = \beta$. In general, $x_1^{(1)},\ldots,x_2^{(L-1)}$ and $x_2^{(1)},\ldots,x_2^{(L-1)}$ are the messages that are \emph{decoded at both receivers}, whereas $x_1^{(L)}, x_2^{(L)}$ are decoded only at their intended receiver (and treated as noise at the other receiver.)

According to Corollary \ref{eqrate}, we show that a finer translation of the power allocation for the messages is achieved by \emph{equalizing the two rate constraints} for every common message ($x_i^{(1)},\ldots,x_i^{(L-1)}, i=1,2$). (However, rates of different common messages are not necessarily the same.)%(i.e., the codes decoded at both receivers in the optimal deterministic scheme.)

As the $1^{st}$ step of determining the power allocations, we give the following lemma on the power allocation of $x_1^{(1)}$ (with the proof found in Appendix \ref{prooflastlem}):
\begin{lem}~ \label{findplem}

1) If $\bar{p}\le \frac{1-g_{12}}{g_{12}^2}$, then $L=1$, and $x_1^{(1)} (x_2^{(1)})$ is treated as noise at the $2^{nd}(1^{st})$ receiver, with $p^{(1)}=\bar{p}$. In this case, there is only one message for each user (as its private message.) rate constraint equalization is \emph{not} needed.

2) If $\bar{p}>\frac{1-g_{12}}{g_{12}^2}$, then $L\ge2$, and $x_1^{(1)} (x_2^{(1)})$ are decoded at both receivers. To \emph{equalize} its rate constraints at both receivers, we must have
\beq \label{findp1}
p^{(1)} = 1 - g_{12} +(1-g_{12}^2)\bar{p} ~~~(<\bar{p}).
\eeq
\end{lem}

Next, we observe that after decoding $x_1^{(1)}, x_2^{(1)}$ at both receivers, determining $p^{(2)}$ for $x_1^{(2)}, x_2^{(2)}$ can be transformed to an equivalent $1^{st}$ step problem with $\bar{p} \leftarrow \bar{p} - p^{(1)}$: solving the new $p^{(1)}$ of the transformed problem gives the correct equalizing solution for $p^{(2)}$ of the original problem. In general, we have the following recursive algorithm in determining $L$ and $p^{(1)},\ldots,p^{(L)}$.

%\vspace{7pt}
%\begin{alg}~

%\begin{itemize}
%\item{Initialize $L = 1$.}
%\item{Step 1: If $\bar{p}\le \frac{1-g_{12}}{g_{12}^2}$, $p^{(L)} \leftarrow \bar{p}$. Terminate.}
%\item{Step 2: $p^{(L)} \leftarrow 1 - g_{12} +(1-g_{12}^2)\bar{p}. ~~L \leftarrow L + 1. ~~\bar{p} \leftarrow \bar{p} - p^{(1)}.$ Return to Step 1.}
%\end{itemize}
%\end{alg}

\begin{tabular}[l]{@{} l @{~} p{12.5cm} @{}} \label{alg2}
\vspace{-14pt}\\
\multicolumn{2}{@{}l}{\emph{Algorithm 2, A finer translation by adapting $L$ and the powers using rate constraint equalization.}}\\
\vspace{-15pt}\\
\hline
\multicolumn{2}{@{}l}{Initialize $L = 1$.}\\

$\quad$&Step 1: If $\bar{p}\le \frac{1-g_{12}}{g_{12}^2}$, then $p^{(L)} \leftarrow \bar{p}$ and terminate.\\

$\quad$&Step 2: $p^{(L)} \leftarrow 1 - g_{12} +(1-g_{12}^2)\bar{p}. ~~L \leftarrow L + 1. ~~\bar{p} \leftarrow \bar{p} - p^{(1)}.$ Go to Step 1.\\
\hline
\\
\end{tabular}

Numerical evaluations of the above simple and finer translations of the optimal scheme of the deterministic channel into that of the Gaussian channel are provided later in Figure \ref{15_30}.

\subsection{Upper Bounds on the Sum-capacity with Successive Decoding of Gaussian Codewords}
In this subsection, we provide two upper bounds on the optimal solution of \eqref{theprobG} for general (asymmetric) channels. More specifically, the bounds are derived for the sum-capacity with Han-Kobayashi schemes, which automatically upper bound the sum-capacity with successive decoding of Gaussian codewords (as shown in Section \ref{HvSsec}.) We will observe that the two bounds have complementary efficiencies, i.e., each being tight in a different regime of parameters.

Similarly to Section \ref{HvSsec}, we denote by $x_i^p$ the private message of the $i^{th}$ user, and $x_i^c$ the common message ($i=1,2$.) We denote $q_i$ to be the power allocated to each \emph{private} message $x_i^p$, $i=1,2$. Then, the power of the common message $x_i^c$ equals $\bar{p}_i - q_i$. WLOG, we normalize the channel parameters such that $g_{11} = g_{22} = 1$. Denote the rates of $x_i^p$ and $x_i^c$ by $r_i^p$ and $r_i^c$. The sum-capacity of Gaussian Han-Kobayashi schemes is thus the following:
\begin{align}
\max_{q_1,q_2}~ & r_1^c + r_1^p + r_2^c + r_2^p\label{maxHK}\\
s.t.~ &\eqref{H1}\sim\eqref{H7}.\nn
\end{align}

To bound \eqref{maxHK}, we select two mutually exclusive \emph{subsets} of the constraints: $\{\eqref{H1},\eqref{H7}\}$ and $\{\eqref{H4}\}$. Then, with each subset of the constraints, a relaxed sum-rate maximization problems can be solved, leading to an \emph{upper bound} to the original constrained sum-capacity \eqref{maxHK}.

The first upper bound on the constrained sum-capacity is as follows (whose proof is immediate from \eqref{H1} and \eqref{H7}):
\begin{lem}
The sum-capacity using Han-Kobayashi schemes is upper bounded by
\begin{align}
opt_1 \triangleq \max_{q_1,q_2}~ \min\big\{~& \log(1+\frac{\bar{p}_1 + g_{21}(\bar{p}_2-q_2)}{g_{21}q_2 + N_1}) + \log(1 + \frac{q_2}{g_{12}q_1 + N_2}), \nn\\
& \log(1+\frac{\bar{p}_2 + g_{12}(\bar{p}_1-q_1)}{g_{12}q_1 + N_2}) + \log(1 + \frac{q_1}{g_{21}q_2 + N_1})~\big\}. \label{UB1}
\end{align}
%where an \emph{analytical solution} to the above max-min problem is available.
\end{lem}
%\begin{proof}

%(1) \emph{Proof of the upper bound}
%\vspace{5pt}
\vspace{3pt}
\begin{proof}[Computation of the upper bound \eqref{UB1}]

%We consider the typical case of $g_{12}\le g_{22} = 1$ and $g_{21}\le g_{11} = 1$.
Note that
\begin{align}
& \log(1+\frac{\bar{p}_1 + g_{21}(\bar{p}_2-q_2)}{g_{21}q_2 + N_1}) + \log(1 + \frac{q_2}{g_{12}q_1 + N_2}) \nn\\
= & ~ \log(c_1) - \log(g_{21}q_2 + N_1) - \log(g_{12}q_1 + N_2) + \log(g_{12}q_1 + q_2 + N_2), \label{minterma}\\
\text{ and }~ & \log(1+\frac{\bar{p}_2 + g_{12}(\bar{p}_1-q_1)}{g_{12}q_1 + N_2}) + \log(1 + \frac{q_1}{g_{21}q_2 + N_1}) \nn\\
= & ~ \log(c_2) - \log(g_{12}q_1 + N_2) - \log(g_{21}q_2 + N_1) + \log(g_{21}q_2 + q_1 + N_1), \label{mintermb}
\end{align}
where $c_1 \triangleq N_1 + \bar{p}_1 + g_{21}\bar{p}_2, c_2 \triangleq N_2 + \bar{p}_2 + g_{12}\bar{p}_1$. Clearly, the minimum of \eqref{minterma} and \eqref{mintermb}) is
\begin{align}
& - \log(g_{21}q_2 + N_1) - \log(g_{12}q_1 + N_2) \label{minoftwo}\\
& + \log\big( \min\{c_1(g_{12}q_1 + q_2 + N_2),c_2(g_{21}q_2 + q_1 + N_1)\} \big). \nn
\end{align}

Now, consider the halfspace $(q_1,q_2)\in \mc{H}$ defined by the linear constraint
\beq
c_1(g_{12}q_1 + q_2 + N_2) \le c_2(g_{21}q_2 + q_1 + N_1) \Leftrightarrow (c_1g_{12} - c_2)q_1 \le (c_2g_{21} - c_1)q_2 + c_2N_1 - c_1N_2. \label{lincons}
\eeq

In $\mc{H}$,
\beq
\eqref{minoftwo} = \log(c_1) - \log(g_{21}q_2 + N_1) - \log(g_{12}q_1 + N_2) + \log(g_{12}q_1 + q_2 + N_2) \triangleq f(q_1,q_2). \label{halfspace1}
\eeq
Note that $\frac{\partial f(q_1,q_2)}{\partial q_1} < 0, \forall q_1\ge 0$. Thus, depending on the sign of $c_1g_{12} - c_2$, we have the following two cases:

\emph{Case 1:} $c_1g_{12} - c_2 \ge 0$. Then, \eqref{lincons} gives an \emph{upper} bound on $q_1$. Consequently, to maximize \eqref{halfspace1}%in $\mc{H}$
, the optimal solution is achieved with $q_1 = 0$. Thus, maximizing \eqref{halfspace1} is equivalent to
\begin{align}
\max_{q_2} &~ - \log(g_{21}q_2 + N_1) + \log(q_2 + N_2) \label{maxq2} \\
\text{s.t. } &~ 0\le q_2\le \bar{p}_2,
\end{align}
in which the objective \eqref{maxq2} is \emph{monotonic}, and the solution is either $q_2 = 0$ or $q_2 = \bar{p}_2$.

\emph{Case 2:} $c_1g_{12} - c_2 < 0$. Then, \eqref{lincons} gives a \emph{lower} bound on $q_1$,
\beq
q_1\ge\frac{(c_1 - c_2g_{21})q_2 + c_1N_2 - c_2N_1}{c_2 - c_1g_{12}}.
\eeq
Consequently, to maximize \eqref{halfspace1}%in $\mc{H}$
, the optimal solution is achieved with $q_1 = \frac{(c_1 - c_2g_{21})q_2 + c_1N_2 - c_2N_1}{c_2 - c_1g_{12}}$, which is a linear function of $q_2$. Substituting this into \eqref{halfspace1}, we need to solve the following problem:
\begin{align}
\max_{q_2} &~ -\log(a_1q_2 + b_1) - \log(a_2q_2 + b_2) + \log(a_3q_2 + b_3) \label{maxq2_2}\\
\text{s.t. } &~ 0\le q_2\le \bar{p}_2, \nn
\end{align}
where $a_i,b_i, (i =1,2,3)$ are constants determined by $c_1,c_2,g_{12},g_{21},N_1,N_2$. Now, \eqref{maxq2_2} can be solved by taking the first derivative w.r.t. $q_2$, and checking the two stationary points and the two boundary points.

In the other halfspace $\mc{H}^c$, the same procedure as above can be applied, and the maximizer of \eqref{minoftwo} within $\mc{H}^c$ can be found. Comparing the two maximizers within $\mc{H}$ and $\mc{H}^c$ respectively, we get the global maximizer of \eqref{UB1}.
\end{proof}

%\vspace{5pt}
The second upper bound on the constrained sum-capacity is as follows (whose proof is immediate from \eqref{H4}):
\begin{lem}
The sum-capacity using Han-Kobayashi schemes is upper bounded by
\begin{align}
opt_2 \triangleq \max_{q_1,q_2}~ \log\left(1+\frac{q_1 + g_{21}(\bar{p}_2 - q_2)}{g_{21}q_2 + N_1}\right) + \log\left(1+\frac{q_2 + g_{12}(\bar{p}_1 - q_1)}{g_{12}q_1 + N_2}\right). \label{UB2}
\end{align}
%where an \emph{analytical solution} to the above maximization problem is available.
\end{lem}

\vspace{3pt}
\begin{proof}[Computation of the upper bound \eqref{UB2}]

Note that
\begin{align}
 &~ \log\left(1+\frac{q_1 + g_{21}(\bar{p}_2 - q_2)}{g_{21}q_2 + N_1}\right) + \log\left(1+\frac{q_2 + g_{12}(\bar{p}_1 - q_1)}{g_{12}q_1 + N_2}\right) \nn\\
= &~ \log(q_1 + g_{21}\bar{p}_2 + N_1) - \log(g_{12}q_1 +N_2) \label{aboutq1}\\
&~ + \log(q_2 + g_{12}\bar{p}_1 + N_2) - \log(g_{21}q_2 +N_1) \label{aboutq2},
\end{align}
where \eqref{aboutq1} is a function only of $q_1$, and \eqref{aboutq2} is a function only of $q_2$. %can be independently optimized, each over $q_1$ and $q_2$ respectively.
Clearly, $\max~\eqref{aboutq1},~s.t.~0\le q_1\le \bar{p}_1$ and $\max~\eqref{aboutq2},~s.t.~0\le q_2\le \bar{p}_2$ can each be solved by taking the first order derivatives, and checking the stationary points and the boundary points.
\end{proof}

\vspace{5pt}
We combine the two upper bounds \eqref{UB1} and \eqref{UB2} as the following theorem:

%\vspace{5pt}
\begin{thm}
The sum-capacity using Gaussian superposition coding-successive decoding is upper bounded by $\min(opt_1, opt_2)$.
\end{thm}
%\vspace{5pt}

\begin{figure}[tb]
  \centering
  \includegraphics[scale = 0.7]{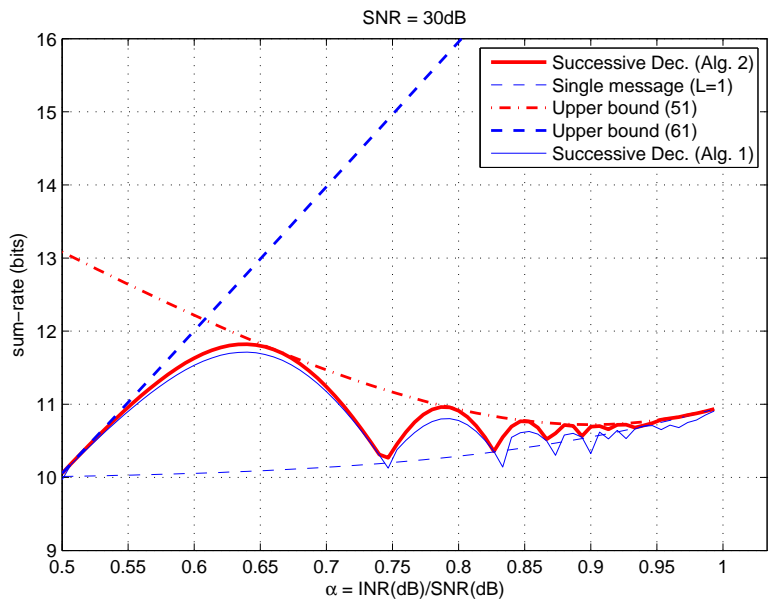}
  \caption{Performance evaluation: achievability vs. upper bounds.}
  \label{15_30}
\end{figure}

\subsection{Performance Evaluation} \label{perf}
We numerically evaluate our results in a symmetric Gaussian interference channel. The $\ms{SNR}$ is set to be $30dB$. To evaluate the performance of successive decoding, we sweep the parameter range of $\alpha = \frac{\log(\ms{INR})}{\log(\ms{SNR})}\in[0.5,1]$, as when $\alpha\in[0,0.5]$, the approximate optimal transmission scheme is simply treating interference as noise without successive decoding.

In Figure \ref{15_30}, the simple translation by Algorithm 1 and the finer translation by Algorithm 2 are evaluated, and the two upper bounds derived above \eqref{UB1}, \eqref{UB2} are computed. The maximum achievable sum-rate with a single message for each user ($L_1=L_2=1$) is also computed, and is used as a baseline scheme for comparison.

We make the following observations:
\begin{itemize}
\item The finer translation of the optimal deterministic scheme by Algorithm 2 is strictly better than the simple translation by Algorithm 1, and is also strictly better than the optimal single message scheme.
\item The first upper bound \eqref{UB1} is tighter for higher $\ms{INR}$ ($\alpha\ge 0.608$ in this example), while the second upper bound \eqref{UB2} is tighter for lower $\ms{INR}$ ($\alpha < 0.608$ in this example).
\item A phenomenon similar to that in the deterministic channels appears: the sum-capacity with successive decoding of Gaussian codewords oscillates between the sum-capacity with Han-Kobayashi schemes and that with single message schemes.
\item The largest difference between the sum-capacity of successive decoding and that of single message schemes appears at around $\frac{\log(\ms{INR})}{\log(\ms{SNR})} = 0.64$, which is about 1.8 bits.
\item The largest difference between the sum-capacity of successive decoding and that of joint decoding (Han-Kobayashi schemes) appears at around $\frac{\log(\ms{INR})}{\log(\ms{SNR})} = 0.74$. This corresponds to the same parameter setting as discussed in Section \ref{HvSsec} (cf. Figure \ref{HKvsSC}). We see that with 30dB $\ms{SNR}$, this largest sum-capacity difference is about 1.0 bits.
\end{itemize}

\begin{figure}[tb]
  \centering
  \includegraphics[scale = 0.7]{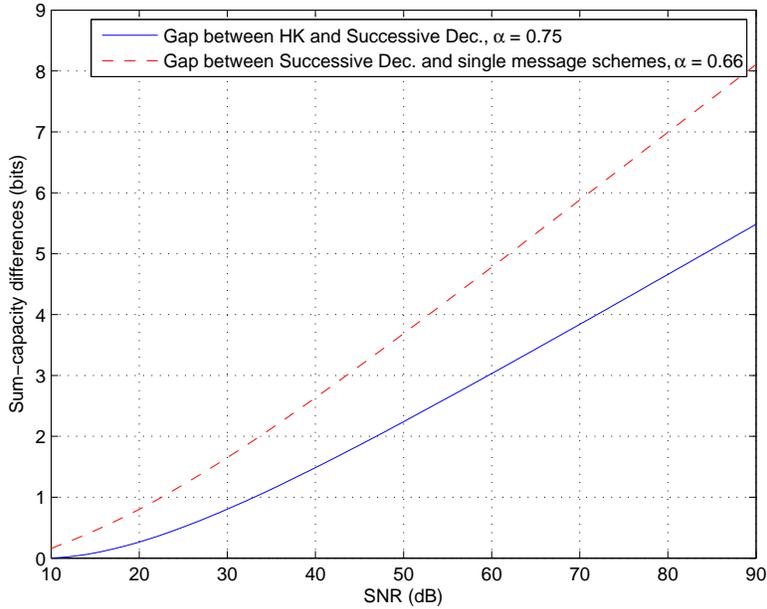}
  \caption{Sum-capacity differences: Han-Kobayashi vs. successive decoding at $\alpha=0.75$, and successive decoding vs. the optimal single message scheme at $\alpha=0.66$.}
  \label{gaps}
\end{figure}

For this particular case with $\ms{SNR} = 30dB$, the observed sum-capacity differences (1.8 bits and 1.0 bits) may not seem very large. However, the capacity curves shown with the deterministic channel model (cf. Figure \ref{siccapafig}) indicate that these differences can go to infinity as $\ms{SNR}\rightarrow\infty$. This is because a rate point $d_{sym}(\alpha)$ on the symmetric capacity curve in the deterministic channel has the following interpretation of \emph{generalized degrees of freedom} in the Gaussian channel \cite{ETW08,BT08}.
\beq
d_{sym}(\alpha) = \lim_{\ms{SNR,INR}\rightarrow\infty, \frac{\log\ms{INR}}{\log\ms{SNR}}=\alpha} \frac{C_{sym}(\ms{INR},\ms{SNR})}{C_{awgn}(\ms{SNR})},
\eeq
where $C_{awgn}(\ms{SNR}) = \log(1+\ms{SNR})$, and $C_{sym}(\ms{INR},\ms{SNR})$ is the symmetric capacity in the two-user symmetric Gaussian channel as a function of $\ms{INR}$ and $\ms{SNR}$.

Since $C_{awgn}(\ms{SNR})\rightarrow\infty$ as $\ms{SNR}\rightarrow\infty$, for a fixed $\alpha$, any finite gap of the achievable rates in the deterministic channel indicates a rate gap that goes to infinity as $\ms{SNR}\rightarrow\infty$ in the Gaussian channel. To illustrate this, we plot the following sum-capacity differences in the Gaussian channel, with $\ms{SNR}$ growing from $10dB$ to $90dB$:
\begin{itemize}
\item The sum-capacity gap between Gaussian superposition coding - successive decoding schemes and single message schemes, with $\alpha = \frac{\log(\ms{INR})}{\log(\ms{SNR})} = 0.66$.
\item The sum-capacity gap between Han-Kobayashi schemes and Gaussian superposition coding - successive decoding schemes, with $\alpha = \frac{\log(\ms{INR})}{\log(\ms{SNR})} = 0.75$.
\end{itemize}

As observed, the sum-capacity gaps increase asymptotically linearly with $\log\ms{SNR}$, and will go to infinity as $\ms{SNR}\rightarrow\infty$.

\section{Concluding Remarks and Discussion} \label{disc}
In this paper, we studied the problem of sum-rate maximization with Gaussian superposition coding and successive decoding in two-user interference channels. This is a hard problem that involves both a combinatorial optimization of decoding orders and a non-convex optimization of power allocation. To approach this problem, we used the deterministic channel model as an educated approximation of the Gaussian channel model, and introduced the complementarity conditions that capture the use of successive decoding of Gaussian codewords. We solved the sum-capacity of the deterministic interference channel with the complementarity conditions, and obtained the capacity achieving schemes with the minimum number of messages. We showed that the constrained sum-capacity oscillates as a function of the cross link gain parameters between the information theoretic sum-capacity and the sum-capacity with interference treated as noise. Furthermore, we showed that if the number of messages used by either of the two users is fewer than its minimum capacity achieving number, the maximum achievable sum-rate drops to that with interference treated as noise. Next, we translated the optimal schemes in the deterministic channel to the Gaussian channel using a rate constraint equalization technique, and provided two upper bounds on the sum-capacity with Gaussian superposition coding and successive decoding. Numerical evaluations of the translation and the upper bounds showed that the constrained sum-capacity oscillates between the sum-capacity with Han-Kobayashi schemes and that with single message schemes.

Next, we discuss some intuitions and generalizations of the coding-decoding assumptions. %future extensions of the work.
\subsection{Complementarity Conditions and Gaussian Codewords}
The complementarity conditions \eqref{cc10}, \eqref{cc20} in the deterministic channel model has played a central role that leads to the discovered oscillating constrained sum-capacity (cf. Theorem \ref{siccapa}). The intuition behind the complementarity conditions is as follows: At any receiver, if two active levels from different users interfere with each other, then \emph{no} information can be recovered at this level. In other words, the \emph{sum of interfering codewords provides nothing helpful}.

This is exactly the case when random Gaussian codewords are used in Gaussian channels with successive decoding, because the sum of two codewords from random Gaussian codebooks cannot be decoded as a valid codeword. This is the reason why the usage of Gaussian codewords with successive decoding is translated to complementarity conditions in the deterministic channels. (Note that the preceding discussions do not apply to \emph{joint} decoding of Gaussian codewords as in Han-Kobayashi schemes.)
\subsection{Modulo-2 Additions and Lattice Codes}
In the deterministic channel, a relaxation on the complementarity conditions is that the \emph{sum} of two interfering active levels can be decoded as their \emph{modulo-2 sum}. As a result, the aggregate of two interfering codewords still provides something valuable that can be exploited to achieve higher capacity. This assumption is part of the original formulation of the deterministic channel model \cite{ADTjour}, with which the information theoretic capacity of the two-user interference channel (cf. Figure \ref{siccapafig} for the symmetric case) can be achieved with Han-Kobayashi schemes \cite{BT08}.

In Gaussian channels, to achieve an effect similar to decoding the modulo-2 sum with successive decoding, \emph{Lattice} codes are natural candidates of the coding schemes. This is because Lattice codebooks have the group property such that the \emph{sum} of two lattice codewords can still be decoded as a valid codeword. Such intermediate information can be decoded first and exploited later during a successive decoding procedure, in order to increase the achievable rate. For this to succeed in interference channels, \emph{alignment of the signal scales} becomes essential \cite{MDFT}. However, our preliminary results have shown that the ability to decode the sum of the Lattice codewords does \emph{not} provide sum-capacity increase for low and medium $\ms{SNR}$s. In the above setting of $\ms{SNR}=30dB$ (which is typically considered as a high $\ms{SNR}$ in practice,) numerical computations show that the sum-capacity using successive decoding of lattice codewords with alignment of signal scales is \emph{lower} than the previously shown achievable sum-rate using successive decoding of Gaussian codewords (cf. Figure \ref{15_30}), for the entire range of $\alpha=\frac{\log{\ms{INR}}}{\log{\ms{SNR}}}\in[0.5,1]$. The reason is that the cost of alignment of the signal scales turns out to be higher than the benefit from it, if $\ms{SNR}$ is not sufficiently high. In summary, no matter using Gaussian codewords or Lattice codewords, the gap between the achievable rate using successive decoding and that using joint decoding can be significant for typical $\ms{SNR}$s in practice. %For successive decoding of \emph{Lattice} codes to achieve similar sum-rate as \emph{joint} decoding schemes, the $\ms{SNR}$ must be even higher.

%\section{Conclusions} \label{conclusec}

\appendix

\subsection{Proof of Lemma \ref{noweak} and \ref{ifnostrong}} \label{prflemlim}
\begin{proof}[Proof of Lemma \ref{noweak}]
By symmetry, it is sufficient to prove for the case $f_2(x) = 1,\forall x\in s_{2i}$, for some $s_{2i}$ that does not end at 1.

Now, consider the sum-rate achieved within $C_1$ \eqref{c1c2}.
As shown in Figure \ref{lemL1}, $C_1$ can be partitioned into three parts: $C_{11} = \{f_1(x)|_{s_1,s_3,\ldots,s_{2i-3}}, f_2(x)|_{s_2,s_4,\ldots,s_{2i-2}}\}$, $C_{12} = \{f_1(x)|_{s_{2i-1},s_{2i+1}}, f_2(x)|_{s_{2i}}\}$, and
$C_{13} = $ $\{f_1(x)|_{s_{2i+3},\ldots}, f_2(x)|_{s_{2i+2},\ldots}\}$, ($C_{11}, C_{12}, C_{13}$ can be degenerate.) %Let $R_{C_{11}}^{sum},R_{C_{12}}^{sum},R_{C_{13}}^{sum}$ denote the sum-rates achieved within $C_{11},C_{12},C_{13}$, and $R_{C_1}^{sum} = R_{C_{11}}^{sum}+R_{C_{12}}^{sum}+R_{C_{13}}^{sum}$.
Note that
\begin{itemize}
\item From the achievable schemes in the proof of Theorem \ref{siccapa}, the maximum achievable sum-rate within $C_{11} \cup C_{13}$ can be achieved with $f_2(x) = 1, \forall x\in s_2\cup s_4 \cup \ldots \cup s_{2i-2} \cup s_{2i+2} \cup \ldots$, and $f_1(x) = 0, \forall x\in s_1\cup s_3 \cup \ldots \cup s_{2i-3} \cup s_{2i+3} \cup\ldots$.
\item By the assumed condition, $f_2(x) = 1,\forall x\in s_{2i} \Rightarrow f_1(x) = 0,\forall x\in s_{2i-1}\cup s_{2i+1}$.
\end{itemize}
Therefore, under the assumed condition, the maximum achievable sum-rate within $C_1$ is achievable with $\{f_2(x) = 1, \forall x\in \mc{G}_2$, and $f_1(x) = 0, \forall x\in \mc{G}_1\}$.

\begin{figure}[t]
  \centering
  \includegraphics[scale = 0.6]{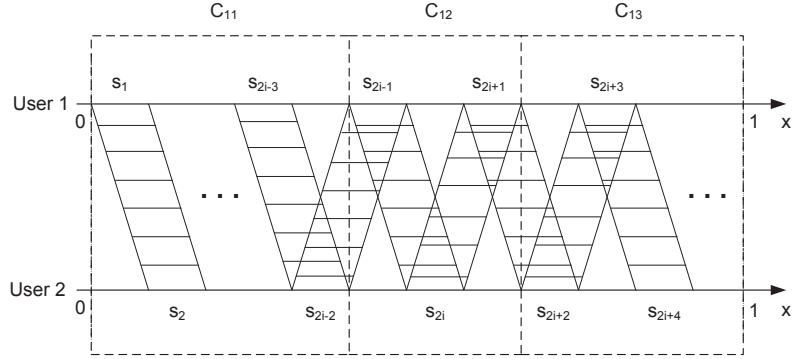}
  \caption{$C_1$ partitioned into three parts for Lemma \ref{noweak}.}
  \label{lemL1}
\end{figure}

Furthermore, from the proof of Theorem \ref{siccapa}, we know that the maximum achievable sum-rate within $C_2$ is achievable with $\{f_2(x) = 1, \forall x\in \mc{G}_1$, and $f_1(x) = 0, \forall x\in \mc{G}_2\}$. Combining the maximum achievable schemes within $C_1$ and $C_2$, by letting $\{f_2(x) = 1, \forall x\in[0,1]$, and $f_1(x) = 0, \forall x\in[0,1]\}$, a sum-rate of 1 is achieved, and this is the maximum achievable sum-rate given the assumed condition.
\end{proof}

\begin{proof}[Proof of Lemma \ref{ifnostrong}]
By symmetry, it is sufficient to prove for the case $f_1(x) = 0,\forall x\in s_{2i-1}$, for some $s_{2i-1}$.

Now, consider the sum-rate achieved within $C_1$.
As shown in Figure \ref{lemL2}, $C_1$ can be partitioned into three parts: $C_{11} = \{f_1(x)|_{s_1,s_3,\ldots,s_{2i-3}}, f_2(x)|_{s_2,s_4,\ldots,s_{2i-2}}\}$, $C_{12} = f_1(x)|_{s_{2i-1}}$, and
$C_{13} = $ $\{f_1(x)|_{s_{2i+1},s_{2i+3},\ldots}, f_2(x)|_{s_{2i},s_{2i+2},\ldots}\}$, ($C_{11}, C_{12}, C_{13}$ can be degenerate.) Note that:
\begin{itemize}
\item From the achievable schemes in the proof of Theorem \ref{siccapa}, the maximum achievable sum-rate within $C_{11} \cup C_{13}$ can be achieved with $f_2(x) = 1, \forall x\in s_2\cup s_4 \cup \ldots \cup s_{2i-2} \cup s_{2i} \cup \ldots$, and $f_1(x) = 0, \forall x\in s_1\cup s_3 \cup \ldots \cup s_{2i-3} \cup s_{2i+1} \cup\ldots$.
\item By the assumed condition, $f_1(x) = 0,\forall x\in s_{2i-1}$.
\end{itemize}
Therefore, under the assumed condition, the maximum achievable sum-rate within $C_1$ is achievable with $\{f_2(x) = 1, \forall x\in \mc{G}_2$, and $f_1(x) = 0, \forall x\in \mc{G}_1\}$.

\begin{figure}[t]
  \centering
  \includegraphics[scale = 0.6]{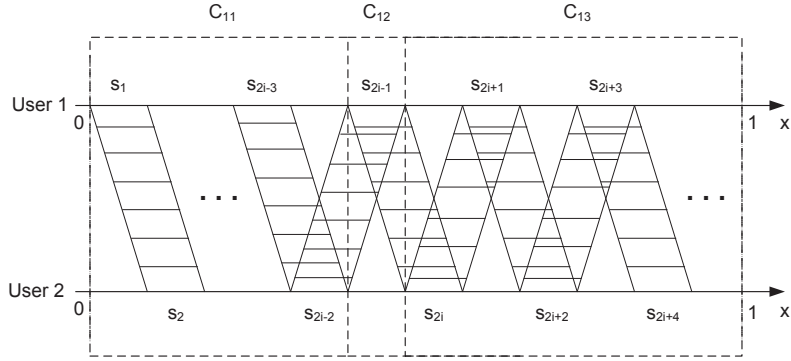}
  \caption{$C_1$ partitioned into three parts for Lemma \ref{ifnostrong}.}
  \label{lemL2}
\end{figure}

Furthermore, from the proof of Theorem \ref{siccapa}, we know that the maximum achievable sum-rate within $C_2$ is achievable with $\{f_2(x) = 1, \forall x\in \mc{G}_1$, and $f_1(x) = 0, \forall x\in \mc{G}_2\}$. Combining the maximum achievable schemes within $C_1$ and $C_2$, by letting $\{f_2(x) = 1, \forall x\in[0,1]$, and $f_1(x) = 0, \forall x\in[0,1]\}$, a sum-rate 1 is achieved, and this is the maximum achievable sum-rate given the assumed condition.
\end{proof}

\subsection{Sum-capacity of Deterministic Asymmetric Interference Channels} \label{asymDC}
In this section, we consider the general two-user interference channel where the parameters $n_{11}, n_{22}, n_{12}, n_{21}$ can be arbitrary. Still, WLOG, we make the assumptions that $n_{11}\ge n_{22}$ and $n_{11} = 1$. We will see that our approaches in the symmetric channel can be similarly extended to solving the constrained sum-capacity in asymmetric channels, without and with constraints on the number of messages.

From Lemma \ref{equivcc}, it is sufficient to consider the following three cases:
\beq
\text{ i) } \delta_1\ge 0 \text{ and } \delta_2\ge 0; ~\text{ ii) } \delta_1\ge 0 \text{ and } \delta_2 < 0; ~\text{ iii) } \delta_1 < 0 \text{ and } \delta_2 \ge 0. \label{asym3case}
\eeq

\vspace{5pt}
\subsubsection{Sum-Capacity without Constraint on the Number of Messages}~ \label{asymwo}
\vspace{5pt}

We provide the optimal scheme that achieves the constrained sum-capacity in each of the three cases in \eqref{asym3case}, respectively.

\vspace{5pt}
\paragraph{$\delta_1\ge 0$ and $\delta_2\ge 0$}~ \label{typcase}
\vspace{5pt}

This is by definition \eqref{deltadef} equivalent to $n_{21} \le 1$ and $n_{22} \ge n_{12}$.

\vspace{5pt}
\emph{Case 1}, $n_{22} \ge n_{21}$:

Define $\beta_1 \triangleq 1 - n_{12}, \beta_2 \triangleq n_{22} - n_{21}$. As depicted in Figure \ref{case11}, interval $I_1 (= [0, 1])$ is partitioned into segments $\{s_1,s_2,s_3,\ldots\}$, with $|s_1| = |s_3| = \ldots = \beta_1$ and $|s_2| = |s_4| = \ldots = \beta_2$; the last segment ending at 1 has the length of the proper residual. Interval $I_2 (= [1 - n_{22}, 1])$ is partitioned into segments $\{s_1',s_2',s_3'\ldots\}$, with $|s_1'| = |s_3'| = \ldots = \beta_2$ and $|s_2'| = |s_4'| = \ldots = \beta_1$; the last segment ending at 1 has the length of the proper residual.

Similarly to \eqref{c1c2} as in the previous analysis for the symmetric channels, we partition the optimization variables $f_1(x)|_{[0,1]}$ and $f_2(x)|_{[1 - n_{22}, 1]}$ into
\beq
C_1\triangleq\{f_1(x)|_{s_1,s_3,\ldots}, f_2(x)|_{s_2',s_4',\ldots}\} \text{ and } C_2\triangleq\{f_1(x)|_{s_2,s_4,\ldots}, f_2(x)|_{s_1',s_3',\ldots}\}. \label{ac1c2}
\eeq

\begin{figure}[h!]
  \centering
  \includegraphics[scale = 0.6]{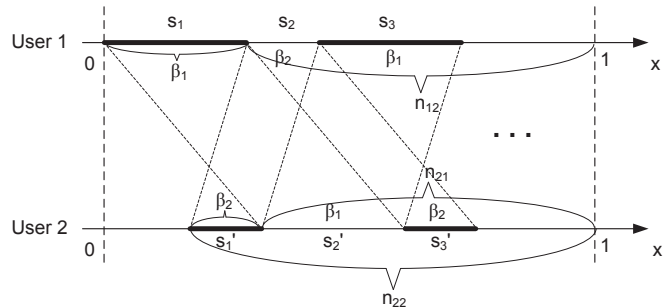}
  \caption{$n_{11} \ge n_{21}$, $n_{22} \ge n_{12}$, and $n_{22} \ge n_{21}$.}
  \label{case11}
\end{figure}

As there is \emph{no constraint between $C_1$ and $C_2$} from the complementarity conditions \eqref{cc10} and \eqref{cc20}, similarly to \eqref{prob1} and \eqref{prob2}, the sum-rate maximization \eqref{theprob0} is decomposed into two separate problems:
\begin{align}
\max_{f_1(x)|_{s_1,s_3,\ldots}, f_2(x)|_{s_2',s_4',\ldots}} & (R_{C_1}^{sum}=)\int_{s_1,s_3,\ldots} f_1(x) \mr{d}x + \int_{s_2',s_4',\ldots} f_2(x) \mr{d}x \label{prob1a} \\
\text{ subject to } & \eqref{indic}, \eqref{cc10}, \eqref{cc20}, \nn
\end{align}
\begin{align}
\max_{f_1(x)|_{s_2,s_4,\ldots}, f_2(x)|_{s_1',s_3',\ldots}} & (R_{C_2}^{sum}=)\int_{s_2,s_4,\ldots} f_1(x) \mr{d}x + \int_{s_1',s_3',\ldots} f_2(x) \mr{d}x \label{prob2a} \\
\text{ subject to } & \eqref{indic}, \eqref{cc10}, \eqref{cc20}. \nn
\end{align}

By the same argument as in the proof of Theorem \ref{siccapa}, the optimal solution of \eqref{prob1a} is given by
\begin{equation}
f_1(x) = 1, \forall x\in s_1\cup s_3\cup\ldots \text{, and } f_2(x) = 0, \forall x\in s_2'\cup s_4'\cup\ldots. \label{achievea1}
\end{equation}
Also, the optimal solution of \eqref{prob2a} is given by
\begin{equation}
f_1(x) = 0, \forall x\in s_2\cup s_4\cup\ldots \text{, and } f_2(x) = 1, \forall x\in s_1'\cup s_3'\cup\ldots. \label{achievea2}
\end{equation}

Consequently, we have the following theorem:
\begin{thm}
A constrained sum-capacity achieving scheme is given by
\begin{equation} \label{achievea}
f_1(x) = \left\{
\begin{array}{rl}
1, & \forall x\in s_1\cup s_3\cup\ldots \\
0, & otherwise \\
\end{array}
\right. , \text{ and }
f_2(x) = \left\{
\begin{array}{rl}
1, & \forall x\in s_1'\cup s_3'\cup\ldots \\
0, & otherwise \\
\end{array}
\right.,
\end{equation}
and the maximum achievable sum-rate is readily computable based on \eqref{achievea}.
\end{thm}

\vspace{5pt}
\emph{Case 2}, $n_{21} > n_{22}$:

Define $\beta_1 \triangleq 1 - n_{12} - (n_{21} - n_{22})$. As depicted in Figure \ref{case12}, interval $I_1 (= [0, 1])$ is partitioned into segments $\{s_0,s_1,s_3,s_5,\ldots\}$, with $|s_0| = n_{21} - n_{22}$, and $|s_1| = |s_3| = \ldots = \beta_1$; the last segment ending at 1 has the length of the proper residual. Interval $I_2 (= [1 - n_{22}, 1])$ is partitioned into segments $\{s_2',s_4'\ldots\}$, with $|s_2'| = |s_4'| = \ldots = \beta_1$; the last segment ending at 1 has the length of the proper residual. (The indexing is not consecutive as we consider $\{s_{2i}\}$ and $\{s_{2i-1}'\}$ ($i\ge1$) as degenerating to empty sets.)

\begin{figure}[h!]
  \centering
  \includegraphics[scale = 0.6]{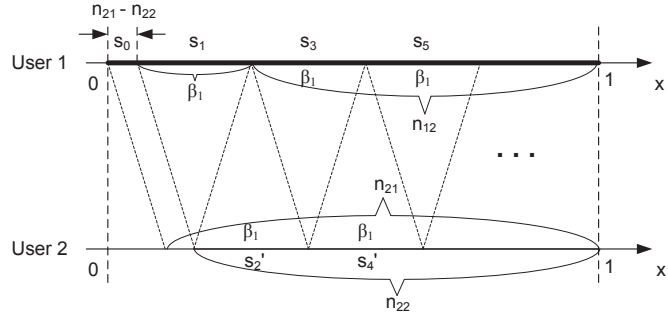}
  \caption{$n_{11} \ge n_{21}$, $n_{22} \ge n_{12}$, and $n_{21} > n_{22}$.}
  \label{case12}
\end{figure}

Clearly, $s_0$ of $I_1$ does not conflict with any levels of $I_2$, and thus we let $f_1(x) = 1, \forall x\in s_0$. On all the other segments, the sum-rate maximization problem is
\begin{align}
\max_{f_1(x)|_{s_1,s_3,\ldots}, f_2(x)|_{s_2',s_4',\ldots}} & \int_{s_1,s_3,\ldots} f_1(x) \mr{d}x + \int_{s_2',s_4',\ldots} f_2(x) \mr{d}x \label{probc12} \\
\text{ subject to } & \eqref{indic}, \eqref{cc10}, \eqref{cc20}. \nn
\end{align}
By the same argument as in the proof of Theorem \ref{siccapa}, the optimal solution of \eqref{probc12} is given by
\begin{equation*}
f_1(x) = 1, \forall x\in s_1\cup s_3\cup\ldots \text{, and } f_2(x) = 0, \forall x\in s_2'\cup s_4'\cup\ldots.
\end{equation*}
Thus, a sum-capacity achieving scheme is simply $f_1(x) = 1, \forall x\in I_1 \text{, and } f_2(x) = 0, \forall x\in I_2$.

\vspace{5pt}
\paragraph{$\delta_1\ge 0$ and $\delta_2 < 0$}~
\vspace{5pt}

This is by definition \eqref{deltadef} equivalent to $n_{21} \le 1$ and $n_{22} < n_{12}$. Note that by Lemma \ref{equivcc}, it is sufficient to only consider the case where $|\delta_1| \ge |\delta_2|$, (because in case $|\delta_1| < |\delta_2|$, we have $|-\delta_2| > |-\delta_1|$.)

\vspace{5pt}
\emph{Case 1}, $n_{22} \ge n_{21}$, and $n_{12} > 1$:

Define $\beta_1 \triangleq n_{22} - n_{21} - (n_{12} -1)$. As depicted in Figure \ref{case21}, interval $I_1 (= [0, 1])$ is partitioned into segments $\{s_1,s_3,\ldots\}$, with $|s_1| = |s_3| = \ldots = \beta_1$; the last segment ending at 1 has the length of the proper residual. Interval $I_2 (= [1 - n_{22}, 1])$ is partitioned into segments $\{s_0',s_2',s_4'\ldots\}$, with $|s_0'| = n_{12} -1$ and $|s_2'| = |s_4'| = \ldots = \beta_1$; the last segment ending at 1 has the length of the proper residual.

\begin{figure}[h!]
  \centering
  \includegraphics[scale = 0.6]{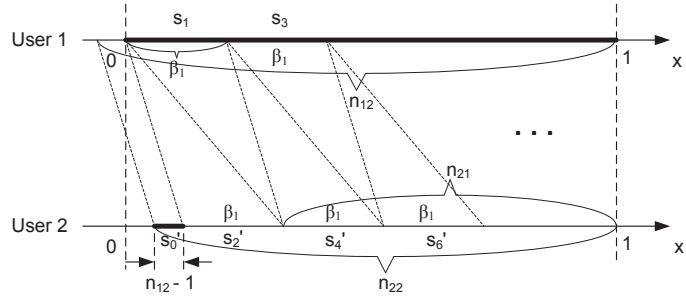}
  \caption{$n_{11} \ge n_{21}$, $n_{22} < n_{12}$, $n_{22} \ge n_{21}$, and $n_{12} > n_{11}$.}
  \label{case21}
\end{figure}

Clearly, $s_0'$ of $I_2$ does not conflict with any levels of $I_1$, and thus we let $f_2(x) = 1, \forall x\in s_0'$. On all the other segments, the sum-rate maximization problem is again \eqref{probc12}, and the optimal solution is given by
\begin{equation*}
f_1(x) = 1, \forall x\in s_1\cup s_3\cup\ldots \text{, and } f_2(x) = 0, \forall x\in s_2'\cup s_4'\cup\ldots.
\end{equation*}
Thus, a sum-capacity achieving scheme is $f_1(x) = 1, \forall x\in I_1 \text{, and } f_2(x) = \left\{
\begin{array}{rl}
1, & \forall x\in s_0'\\
0, & otherwise \\
\end{array}
\right.$.

\vspace{5pt}
\emph{Case 2}, $n_{22} \ge n_{21}$, and $n_{12} \le 1$:

Define $\beta_1 \triangleq 1 - n_{12}, \beta_2 \triangleq n_{22} - n_{21}$. As depicted in Figure \ref{case221}, interval $I_1 (= [0, 1])$ is partitioned into segments $\{s_1,s_2,s_3,\ldots\}$, with $|s_1| = |s_3| = \ldots = \beta_1$ and $|s_2| = |s_4| = \ldots = \beta_2$; the last segment ending at 1 has the length of the proper residual. Interval $I_2 (= [1 - n_{22}, 1])$ is partitioned into segments $\{s_1',s_2',s_3'\ldots\}$, with $|s_1'| = |s_3'| = \ldots = \beta_2$ and $|s_2'| = |s_4'| = \ldots = \beta_1$; the last segment ending at 1 has the length of the proper residual.

\begin{figure}[h!]
  \centering
  \includegraphics[scale = 0.6]{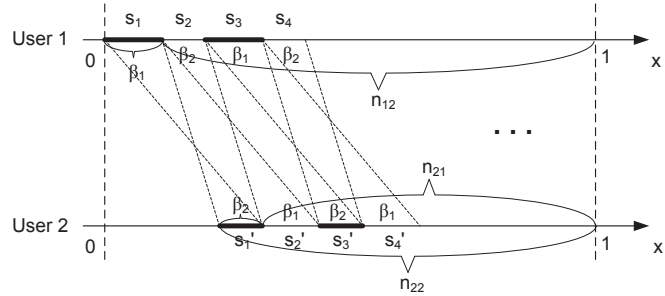}
  \caption{$n_{11} \ge n_{21}$, $n_{22} < n_{12}$, $n_{22} \ge n_{21}$, and $n_{12} \le n_{11}$, scheme I (non-optimal).}
  \label{case221}
\end{figure}

Compare with Case 1 of Section \ref{typcase} and note the similarities between Figure \ref{case221} and Figure \ref{case11}: we apply the same partition of the optimization variables \eqref{ac1c2}, and the sum-rate maximization \eqref{theprob0} is decomposed in the same way into two separate problems \eqref{prob1a} and \eqref{prob2a}. However, while the optimal solution of \eqref{prob1a} is still given by \eqref{achievea1}, \emph{the optimal solution of \eqref{prob2a} is no longer given by \eqref{achievea2}}. Instead, as $\delta_2<0$, the optimal solution of \eqref{prob2a} is given by
\begin{equation*}
f_1(x) = 1, \forall x\in s_2\cup s_4\cup\ldots \text{, and } f_2(x) = 0, \forall x\in s_1'\cup s_3'\cup\ldots.
\end{equation*}
Thus, a sum-capacity achieving scheme is given by $f_1(x) = 1, \forall x\in I_1 \text{, and } f_2(x) = 0, \forall x\in I_2$, depicted as in Figure \ref{case222}.

\begin{figure}[h!]
  \centering
  \includegraphics[scale = 0.6]{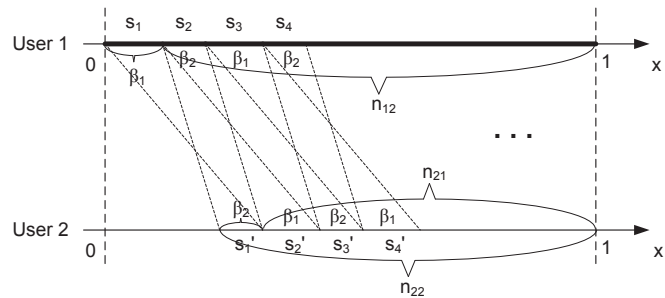}
  \caption{$n_{11} \ge n_{21}$, $n_{22} < n_{12}$, $n_{22} \ge n_{21}$, and $n_{12} \le n_{11}$, scheme II (optimal).}
  \label{case222}
\end{figure}

\vspace{5pt}
\emph{Case 3}, $n_{22} < n_{21}$:

Compare with Case 2 of \ref{typcase} (cf. Figure \ref{case12}), with the same definition of $\beta_1$ and the same partition of $I_1$ and $I_2$, the segmentation is depicted in Figure \ref{case23}.

\begin{figure}[h!]
  \centering
  \includegraphics[scale = 0.6]{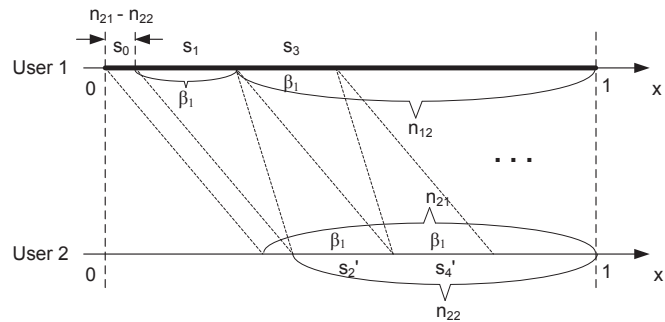}
  \caption{$n_{11} \ge n_{21}$, $n_{22} < n_{12}$, and $n_{22} < n_{21}$.}
  \label{case23}
\end{figure}

Noting the similarities between Figure \ref{case12} and Figure \ref{case23}, we see that the optimal solution of the two cases are the same: $f_1(x) = 1, \forall x\in I_1 \text{, and } f_2(x) = 0, \forall x\in I_2$.

\vspace{5pt}
\paragraph{$\delta_1 < 0$ and $\delta_2 \ge 0$}~
\vspace{5pt}

This is by definition \eqref{deltadef} equivalent to $n_{21} > 1$ and $n_{22} \ge n_{12}$. Note that by Lemma \ref{equivcc}, it is sufficient to only consider the case where $|\delta_1| \le |\delta_2|$, (because in case $|\delta_1| > |\delta_2|$, we have $|-\delta_2| \le |-\delta_1|$.)

Define $\beta_1 \triangleq 1 - n_{12} - (n_{21} - n_{22})$. As depicted in Figure \ref{case3}, interval $I_1 (= [0, 1])$ is partitioned into segments $\{s_0,s_1,s_3,s_5,\ldots\}$, with $|s_0| = n_{21} - n_{22}$ and $|s_1| = |s_3| = \ldots = \beta_1$; the last segment ending at 1 has the length of the proper residual. Interval $I_2 (= [1 - n_{22}, 1])$ is partitioned into segments $\{s_2',s_4'\ldots\}$, with $|s_2'| = |s_4'| = \ldots = \beta_1$; the last segment ending at 1 has the length of the proper residual.

\begin{figure}[h!]
  \centering
  \includegraphics[scale = 0.6]{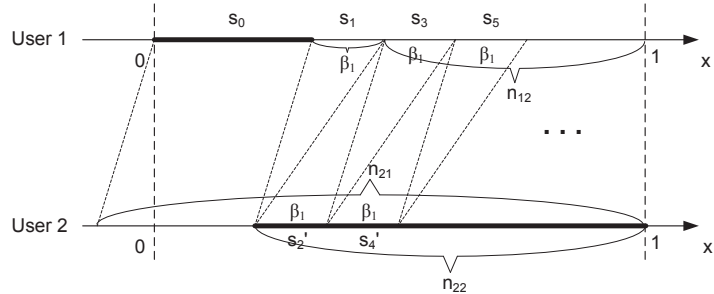}
  \caption{$n_{11} < n_{21}$ and $n_{22} \ge n_{12}$.}
  \label{case3}
\end{figure}

Clearly, $s_0$ of $I_1$ does not conflict with any levels of $I_2$, and thus we let $f_1(x) = 1, \forall x\in s_0$. On all the other segments, the sum-rate maximization problem is again \eqref{probc12}. As $\delta_1<0$, the optimal solution is given by
\begin{equation*}
f_1(x) = 0, \forall x\in s_1\cup s_3\cup\ldots \text{, and } f_2(x) = 1, \forall x\in s_2'\cup s_4'\cup\ldots.
\end{equation*}
Thus, a sum-capacity achieving scheme is
$f_1(x) = \left\{
\begin{array}{rl}
1, & \forall x\in s_0\\
0, & otherwise \\
\end{array}
\right. \text{, and } f_2(x) = 1, \forall x\in I_2.$

\vspace{5pt}
Summarizing the discussions of the six parameter settings (cf. Figures \ref{case11} - \ref{case21} and \ref{case222} - \ref{case3}) in this subsection, we observe:
\begin{RK} \label{numcodesasym}
Except for Case 1 of Section \ref{typcase}, the optimal schemes for the other cases all have the property that \emph{only one message is used for each user}.
\end{RK}

\vspace{10pt}
\subsubsection{The Case with a Limited Number of Messages}~
\vspace{5pt}

In this subsection, we extend the sum-capacity results in Section \ref{limsym} to the asymmetric channels when there are upper bounds on the number of messages $L_1,L_2$ for the two users respectively.
From Remark \ref{numcodesasym}, we only need to discuss Case 1 of Section \ref{typcase} (cf. Figure \ref{case11},) with its corresponding notations.

Similarly to the symmetric channels, we generalize Lemma \ref{noweak} and \ref{ifnostrong} to the following two lemmas for the general (asymmetric) channels, whose proofs are exact parallels to those of Lemma \ref{noweak} and \ref{ifnostrong}:

\begin{lem}~ \label{noweaka}

1. If $\exists s_{2i}$, $s_{2i}$ does not end at 1, such that $f_1(x) = 1,\forall x\in s_{2i}$, then $R^{sum} \le 1$.

2. If $\exists s_{2i}'$, $s_{2i}'$ does not end at 1, such that $f_2(x) = 1,\forall x\in s_{2i}'$, then $R^{sum} \le n_{22}$.
\end{lem}
\begin{lem} \label{ifnostronga}~

1. If $\exists s_{2i-1}$, such that $f_1(x) = 0,\forall x\in s_{2i-1}$, then $R^{sum} \le n_{22}$.

2. If $\exists s_{2i-1}'$, such that $f_2(x) = 0,\forall x\in s_{2i-1}'$, then $R^{sum} \le 1$.
\end{lem}
We then have the following generalization of Theorem \ref{lessL} to the general (asymmetric) channels:

\vspace{5pt}
\begin{thm}\label{limasym}
Denote by $L_i$ the number of messages used by the $i^{th}$ user in any scheme, and denote by $n_i$ the dictated number of messages used by the $i^{th}$ user in the constrained sum-capacity achieving scheme \eqref{achievea}. Then, if $L_1\le n_1-1$ \emph{or} $L_2\le n_2-1$, we have $R^{sum} \le 1$.
\end{thm}
\begin{proof}
Consider $L_2\le n_2-1$. (The case of $L_1\le n_1-1$ can be proved similarly.)

i) The sum-rate of 1 is always achievable with
\begin{equation*}
f_1(x) = 1, \forall x\in I_1, f_2(x) = 0, \forall x\in I_2.
\end{equation*}

ii) If there exists $s_{2i}', (i\ge 1)$ and $s_{2i}'$ does not end at 1, such that $f_2(x) = 1,\forall x\in s_{2i}'$, then from Lemma \ref{noweaka}, $R^{sum}\le n_{22}\le 1$.

iii) If \emph{for every} $s_{2i}', i\ge 1$ and $s_{2i}'$ does not end at 1, there exists $x_{i}$ in the interior of $s_{2i}'$ such that $f_2(x_{i})= 0$:

For every $x_{i}$, since $s_{2i}'$ does not end at 1, $s_{2i+1}'$ exists. Note that $x_{i}$ \emph{separates} the two segments $s_{2i-1}', s_{2i+1}'$ for the $2^{nd}$ user. From Remark \ref{numcodes}, $s_{2i-1}'$ and $s_{2i+1}'$ have to be \emph{two distinct messages} provided that both of them are (at least partly) active for the $2^{nd}$ user. On the other hand, there are $n_2$ such segments $\{s_1',s_3',\ldots,s_{2n_2-1}'\}$, whereas the number of messages is upper bounded by $L_2\le n_2-1$. Consequently, $\exists 1\le i_2\le n_2$, such that $f_2(x) = 0, \forall x\in s_{2i_2-1}$. In other words, for the $2^{nd}$ user, there must be a segment with an odd index that is \emph{fully inactive}. By Lemma \ref{ifnostronga}, in this case, $R^{sum}\le 1$.
\end{proof}

Similarly to the symmetric case, we conclude that if the number of messages used for \emph{either} user is fewer than the number used in the optimal scheme \eqref{achievea}, the maximum achievable sum-rate drops to $1$.

\subsection{Proof of Lemma \ref{findplem}} \label{prooflastlem}
At the $1^{st}$ receiver, the message $x_1^{(1)}$ is decoded by treating all other messages ($x_1^{(2)},\ldots,x_1^{(L)},x_2^{(1)},\ldots,x_2^{(L)}$) as noise, and has an $\ms{SNR_1}$ of $\frac{p^{(1)}}{(\bar{p} - p^{(1)}) + g_{21}\bar{p} + 1}$.

At the $2^{nd}$ receiver, $x_2^{(1)}$ is first decoded and peeled off. Suppose $x_1^{(1)}$ is also decoded at the $2^{nd}$ receiver (by treating $x_1^{(2)},\ldots,x_1^{(L)},$ $x_2^{(2)},\ldots,x_2^{(L)}$ as noise,) it has an $\ms{SNR_2}$ of $\frac{g_{12}p^{(1)}}{g_{12}(\bar{p} - p^{(1)}) + (\bar{p} - p^{(1)}) + 1}$. To equalize the rate constraints for $x_1^{(1)}$ at both receivers, we need
\begin{equation*}
\ms{SNR_1}=\ms{SNR_2} \Rightarrow p^{(1)} = 1 - g_{12} +(1-g_{12}^2)\bar{p}.
\end{equation*}
Note that $p^{(1)}<\bar{p}$ requires that $\bar{p}>\frac{1-g_{12}}{g_{12}^2}$.
Otherwise, $\bar{p}\le\frac{1-g_{12}}{g_{12}^2}$, and the above $1 - g_{12} +(1-g_{12}^2)\bar{p}\ge \bar{p}$. It implies that we should not decode $x_1^{(1)}$ at the $2^{nd}$ receiver, i.e., $x_i^{(1)} (i=1,2)$ is the only message ($L=1$) of the $i^{th}$ user, which is treated as noise at the other receiver.

% use section* for acknowledgement
% \section*{Acknowledgment}

% The authors would like to thank...

% trigger a \newpage just before the given reference
% number - used to balance the columns on the last page
% adjust value as needed - may need to be readjusted if
% the document is modified later
%\IEEEtriggeratref{8}
% The "triggered" command can be changed if desired:
%\IEEEtriggercmd{\enlargethispage{-5in}}

% references section

% can use a bibliography generated by BibTeX as a .bbl file
% BibTeX documentation can be easily obtained at:
% http://www.ctan.org/tex-archive/biblio/bibtex/contrib/doc/
% The IEEEtran BibTeX style support page is at:
% http://www.michaelshell.org/tex/ieeetran/bibtex/
%\bibliographystyle{IEEEtran}
% argument is your BibTeX string definitions and bibliography database(s)
%\bibliography{IEEEabrv,../bib/paper}
%
% <OR> manually copy in the resultant .bbl file
% set second argument of \begin to the number of references
% (used to reserve space for the reference number labels box)

\bibliographystyle{plain}
{\bibliography{bare_conf}}

%\begin{thebibliography}{1}

%\bibitem{}

%\end{thebibliography}

% that's all folks
\end{document}